%% file: main.tex
\documentclass[sigconf,nonacm]{acmart}

\usepackage{enumitem}
\setlist[itemize]{leftmargin=*}
\setlist[enumerate]{leftmargin=*}
\usepackage{algorithmicx}
\usepackage{algpseudocode}

\usepackage{braket}
\usepackage{algorithm}
\usepackage{algpseudocode}
\floatname{algorithm}{Protocol}
\usepackage{amsmath,amsthm,amsfonts}
\usepackage{booktabs}
\usepackage{array}
\usepackage{makecell}
\usepackage{multirow}
\newcolumntype{C}[1]{>{\centering\arraybackslash}m{#1}}
\usepackage{mathtools}
\usepackage{thmtools}
\usepackage{thm-restate}
\usepackage{tabularx}
\usepackage{bm}
\usepackage[dvipsnames]{xcolor}
\newtheorem{theorem}{Theorem}
\newtheorem{definition}[theorem]{Definition}
\newtheorem{proposition}[theorem]{Proposition}
\newtheorem{lemma}[theorem]{Lemma}

\AtBeginDocument{%
  }

\title{OQRAM: Oblivious Quantum Random Access Memory for Securing Delegated Quantum Queries}

\author{Shifan Xu}
\authornote{Shifan Xu and Yizhuo Tan contributed equally to this work.}
\affiliation{%
  \institution{Yale University}
  \city{New Haven}
  \state{Connecticut}
  \country{USA}}
\email{shifan.xu@yale.edu}

\author{Yizhuo Tan}
\authornotemark[1]
\affiliation{%
  \institution{Yale University}
  \city{New Haven}
  \state{Connecticut}
  \country{USA}}
\email{yizhuo.tan@yale.edu}

\author{Yongshan Ding}
\affiliation{%
  \institution{Yale University}
  \city{New Haven}
  \state{Connecticut}
  \country{USA}}
\email{yongshan.ding@yale.edu}

\author{Jakub Szefer}
\affiliation{%
  \institution{Northwestern University}
  \city{Evanston}
  \state{Illinois}
  \country{USA}}
\email{jakub.szefer@northwestern.edu}

\renewcommand{\shortauthors}{Anonymous Author(s)}

\begin{abstract}
Quantum query is a basic subroutine in many quantum algorithms, and Quantum Random Access Memory (QRAM) provides a natural way to realize such coherent query access. In delegated settings, however, a standard QRAM query interface can expose sensitive information to the server. This paper introduces oblivious QRAM, a cryptographic abstraction for privacy-preserving delegated coherent query access. The protocol consists of an offline refresh phase and an online protected query phase. The database is stored in an encrypted and shuffled layout, and each query is protected by coherent address masking using either a quantum-secure pseudorandom permutation (qPRP) based method or a quantum one-time pad (qOTP) based method. In the adopted client model, the online protection adds only modest quantum overhead beyond the query register, avoiding the exponential quantum resources that would otherwise be required by an equivalent local QRAM construction. The qPRP-based variant also supports multi-query use by distributing database refresh across multiple queries to reduce classical communication. To address malicious servers, decoy checks are further incorporated to strengthen privacy protection and enable probabilistic tampering detection. Compared with fully blind quantum computing, this framework provides a lighter abstraction tailored to private delegated QRAM access, significantly reducing quantum resource requirements on both the client and server sides and achieving an exponential reduction in quantum communication.
\end{abstract}


\begin{document}
\maketitle

\input{introduction}

\input{background}

\input{threatmodel}

\input{Protocol}

\input{variant}

\input{proof}

\input{results}

\input{conclusion}

\begin{acks}
We thank Ang Li and Zhixin Song for fruitful discussion. This paper was edited for grammar using ChatGPT. This material is based upon work supported by National Science Foundation (under awards CCF-2312754 and CCF-2332406). YD acknowledges partial support by NSF CCF-2338063, by the U.S. Department of Energy, Office of Science, National Quantum Information Science Research Center, Co-design Center for Quantum Advantage (C2QA) under Contract No. DE-SC0012704, by Quantum CT (under NSF Engines award ITE-2302908), by Air Force Office of Scientific Research MURI (FA9550-26-1-B036), by Boehringer Ingelheim, and NSF NQVL-ERASE (under award OSI-2435244), and by DARPA under award HR0011-26-9-E123. External interest disclosure: YD is a consultant and equity holder of D-Wave Quantum, Inc.

\end{acks}

\bibliographystyle{plain}
\bibliography{refs}

\appendix
\section{Open Science}
This paper is primarily a theoretical and analytical contribution. The proposed protocols, threat models, security arguments, and resource estimates are fully described in the paper. The submission does not include external artifacts such as code, datasets, benchmark suites, or executable evaluation scripts. No external artifact is necessary to evaluate the paper beyond the definitions, constructions, proofs, and asymptotic cost analyses included in the manuscript.

\section{Ethical Considerations}
This work is theoretical and does not involve human-subject data, private datasets, or deployed systems. The adversarial models are included to define privacy and integrity risks in delegated quantum query access and to motivate defensive protocol design. The paper does not provide an attack implementation; potential misuse is limited to conceptual insights, which are presented together with mitigations and leakage limits.

\input{appendix}

\end{document}

%% file: introduction.tex
\section{Introduction}
\label{sec:intro}

Outsourced memory access is a basic setting in classical cloud security: a client relies on a remote server to store or access a large dataset, while the access interface itself may reveal information even when the stored data are encrypted~\cite{goldreich1996software,stefanov2013path,islam2012access,kellaris2016generic}. Cloud quantum computation brings this outsourced-memory setting into the quantum regime, especially in scenarios where a limited-quantum-resource client delegates specialized subroutines to a remote server with substantially more powerful quantum hardware~\cite{childs2005secure,broadbent2009universal,fitzsimons2017private}. Among these subroutines, memory access is particularly important. Many quantum algorithms are naturally expressed in the query model, where the computational advantage comes from repeated oracle access rather than from a fully unrolled circuit description~\cite{beals2001quantum,grover1996fast,simon1997power,childs2003exponential}. In practice, a standard way to realize such coherent access to large classical data is Quantum Random Access Memory (QRAM), which allows a quantum processor to query a large outsourced dataset while keeping the query in quantum superposition~\cite{giovannetti2008quantum,giovannetti2008architectures,xu2023systems,xu2025fattree,wang2025transmon,weiss2024quantum,weiss2024faulty}. As a result, delegated QRAM is a natural primitive for cloud-based quantum computing.

From a security perspective, delegated QRAM should be viewed as an outsourced memory service with a much richer interface than ordinary RAM. In a classical setting, a memory query reveals a single concrete address unless additional mechanisms such as Oblivious RAM (ORAM) are used to hide the access pattern~\cite{goldreich1987oram,goldreich1996software,stefanov2013path}. In delegated QRAM, however, the query itself may be in a coherent superposition of addresses. This means that the “address” is no longer merely a single record index. Instead, it may encode sensitive algorithmic state in its amplitudes and phases, such as a search predicate, a hidden hypothesis, or algorithmic choices of the surrounding quantum algorithm, none of which would appear as a classical memory address. Consequently, even if the outsourced database may be encrypted, a standard delegated QRAM interface can still leak information through the query state itself.

This leakage is not specific to a poorly designed QRAM implementation. It is inherent in the usual delegated QRAM interface, which requires the client to hand a quantum address state to the server so that the server can route the query through its memory system quantumly. If the server can inspect, perturb, or correlate that state across rounds, then the privacy of the surrounding algorithm may be compromised at the query interface alone. This makes the delegated QRAM call itself a first-class attack surface.

Existing notions, while closely related, do not directly address this problem. Prior work differs along three main dimensions: whether the queried address is classical or quantum, whether the interface returns a single record or implements coherent unitary access, and whether the security goal is query privacy or full delegated-computation blindness.
Private information retrieval and quantum private query protocols typically hide which classical record is retrieved from a database, often with information-theoretic or cheat-sensitive guarantees, but they do not provide a mechanism for preserving an arbitrary coherent address superposition as the input to a delegated unitary QRAM operation~\cite{chor1998pir,kerenidis2004qspir,giovannetti2008quantum,jakobi2011private}.
Blind quantum computation (BQC) hides a much richer delegated computation, including the client’s input, circuit, and output, through interactive protocols based on tools such as randomized state preparation and hidden measurement angles~\cite{broadbent2009universal,fitzsimons2017private,li2023robust,bourdoncle2025practical}. This generality is powerful, but substantially heavier than needed for repeated memory access.
Classical ORAM hides access patterns over time by remapping, reshuffling, and re-encrypting outsourced storage, and ORAM-style techniques are useful for layout refresh and sparse classical updates. However, ORAM remains a classical hidden-access mechanism and does not by itself address a server that receives a coherent quantum address state~\cite{goldreich1996software,stefanov2013path,ren2015ring,mishra2018oblix,sasy2018zerotrace,asharov2020optorama,tople2018prooram,gagliardoni2017orams}.

These frameworks therefore leave a gap between classical hidden-access protocols and full blind quantum computation. This gap motivates a dedicated notion of secure quantum query, where the server executes a QRAM lookup without learning the client's logical address state and the client retains a lightweight query-oriented interface. This setting is especially relevant in the asymmetric regime considered in this work, where the client has only \(O(n)\) reliable qubits, while the server supports a QRAM of size \(N=2^n\) using \(O(N)\) hardware resources. In this regime, the client cannot simply replicate the memory system locally, yet the privacy of the address register may remain essential to the surrounding algorithm. The central question is therefore whether delegated unitary query access can be protected without paying the cost of full blind delegated quantum computation.

This work answers that question by introducing oblivious QRAM, a cryptographic framework for privacy-preserving delegated coherent query access. The main contributions are as follows.
\begin{itemize}
    \item A security abstraction for delegated quantum query clarifies the distinction between honest-but-curious privacy and malicious-server robustness.
    \item A baseline oblivious QRAM protocol, combining an encrypted shuffled database with coherent protection of each online quantum query.
    \item Extensions of the protocol, including a qOTP-based single query variant, two-round query variant, and decoy-based checks for probabilistic tampering detection.
    \item A resource and deployment analysis, covering online/offline quantum and classical resource estimation, refresh mechanisms, and reshuffling strategies.
\end{itemize}

The rest of the paper is organized as follows. Sec.~\ref{sec:background} reviews the QRAM query model and the cryptographic tools used in the construction. Sec.~\ref{sec:threatmodel} defines the system assumptions, adversarial assumptions, security goals, and leakage scope. Sec.~\ref{sec:protocol} presents the main oblivious QRAM protocol, including the offline refresh phase, the online protected query phase, and the arbitrary-bus extension. Sec.~\ref{sec:variants} describes protocol variants and malicious-server protections. Sec.~\ref{sec:security_proof} discusses the main security claims and proof intuition. Sec.~\ref{sec:results} analyzes resource costs and deployment options, including refresh and synchronization strategies. Finally, Sec.~\ref{sec:conclusion} concludes.

%% file: background.tex
\section{Background}
\label{sec:background}

\begin{figure}[t]
    \centering
    \includegraphics[width=\linewidth]{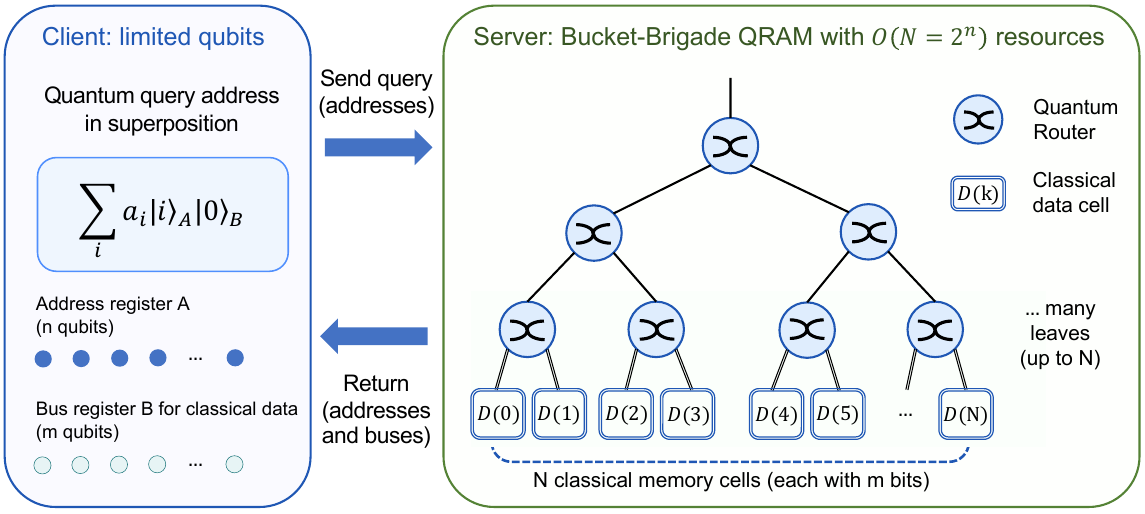}
    \caption{Delegated quantum query with a bucket-brigade QRAM server. A lightweight client prepares a coherent address state and sends the address register to a server that implements a bucket-brigade QRAM with \(O(N)\) memory and routing resources. The QRAM coherently routes the superposed address through a tree-like structure and returns both address and bus registers.}
    \label{fig:bg}
\end{figure}

\subsection{Quantum Query and QRAM}
\label{subsec:quantum-query}

Let \(D\) be a classical database indexed by \(n\)-bit addresses \(i\in\{0,1\}^n\), where each record \(D[i]\in\{0,1\}^m\) is loaded into an \(m\)-qubit bus register. The basic QRAM lookup on a single computational-basis address is similar to a classical query, written as the reversible block-value map
\begin{equation}
Q_D:\ |i\rangle_A|b\rangle_B \longmapsto |i\rangle_A|b\oplus D[i]\rangle_B,
\label{eq:classical-query}
\end{equation}
where \(i\in\{0,1\}^n\), \(b\in\{0,1\}^m\), and \(\oplus\) is bitwise XOR on the bus register. When \(B\) is initialized to \(|0^m\rangle_B\), this returns \(|D[i]\rangle_B\). A quantum query uses the linear extension of the same lookup when the address register is prepared in a superposition of addresses. In particular, for an arbitrary address state
\(|\psi\rangle_A=\sum_i \alpha_i |i\rangle_A,\)
which may contain all addresses in the database, the QRAM returns
\begin{equation}
Q_D\left(\sum_i\alpha_i\ket{i}_A\ket{0^m}_B\right)
=
\sum_i\alpha_i\ket{i}_A\ket{D[i]}_B .
\label{eq:qram-superposition}
\end{equation}

Architecturally, QRAM is usually modeled as a large tree-like qubit array with over \(N=2^n\) qubits attached to a classical memory of size \(N\) at the leaves of that tree~\cite{giovannetti2008architectures,giovannetti2008quantum}. As shown in Figure~\ref{fig:bg}, the large qubit overhead of QRAM motivates the quantum resource asymmetry considered in this work. The client needs only the \(n\)-qubit address register and a small number of local work qubits, while the server provides the \(O(N)\)-scale memory and routing hardware required to realize the query.

\subsection{Quantum Secure Pseudorandom Permutations}
A quantum-secure pseudorandom permutation (qPRP)~\cite{luby1988prp,zhandry2025note} is a keyed family of efficiently computable permutations
\[
\{\pi_K:\{0,1\}^n\to\{0,1\}^n\}_{K\in\mathcal{K}},
\]
where \(K\in\mathcal{K}\) is the permutation key, \(\mathcal{K}\) is the key space, and \(n\) is the input length. Each permutation \(\pi_K\) induces a unitary action on computational-basis states:
\[
U_{\pi_K} |i\rangle = |\pi_K(i)\rangle, \qquad
U_{\pi_K}^{\dagger} |j\rangle = |\pi_K^{-1}(j)\rangle.
\]
A \emph{strong} quantum-secure pseudorandom permutation remains indistinguishable from a uniformly random permutation even to a quantum polynomial-time adversary with superposition access to both the permutation and its inverse~\cite{zhandry2025note}. This is the relevant notion for delegated QRAM because the server may process quantum query states coherently and may interact with both forward and inverse permutation structures across the protocol.

A natural construction of a qPRP is a balanced seven-round Feistel network over \(\{0,1\}^n\)~\cite{carolan2025compressed}. Writing the input as \(L\|R \in \{0,1\}^{n/2}\times\{0,1\}^{n/2}\), round \(i\) applies
\[
(L,R)\mapsto (R,\,L\oplus F_i(R)),
\]
where \(F_i:\{0,1\}^{n/2}\to\{0,1\}^{n/2}\) is the \(i\)-th keyed round function, and the full permutation is given by the composition of seven such rounds.

Importantly, there are two levels of abstraction in this instantiation. At the theorem level, the security analysis~\cite{carolan2025compressed} models the round functions \(F_i\) as independent ideal random functions and shows that the resulting seven-round balanced Feistel construction achieves strong qPRP security in the bidirectional quantum query model, with distinguishing advantage \(O(q^3 / N^{1/4})\) over a domain of size \(N\), where \(q\) is the total number of forward and inverse quantum queries. The formal qPRP security statement is used later in the protocol security discussion in Sec.~\ref{sec:security_proof}. At the implementation level, the idealized round functions must be instantiated with concrete, efficiently computable primitives, since truly random functions are only an abstraction used in the security proof. A natural candidate is to instantiate each \(F_i\) using a qPRF under a post-quantum assumption such as LWE, thereby yielding an efficiently invertible permutation. More broadly, prior works show that strong qPRPs can be constructed from qPRFs together with classical format-preserving encryption techniques, and hence from quantum-resistant one-way functions~\cite{zhandry2025note}.

\subsection{Quantum Secure Symmetric Encryption}

For outsourced databases, encryption is needed to protect the contents of the stored records against an untrusted quantum server~\cite{song2000encryptedsearch,curtmola2006sse}. To protect the confidentiality of stored data against a quantum adversary, one typically considers quantum indistinguishability under quantum chosen-plaintext attack, or qIND-qCPA security~\cite{boneh2013quantum,gagliardoni2016semantic}. Informally, this notion requires that no quantum polynomial-time adversary with quantum chosen-plaintext access can distinguish encryptions of chosen messages except with negligible advantage.

Prior work proposes a concrete qIND-qCPA secure construction~\cite{gagliardoni2016semantic}. Let $\Pi_{m+\tau}=(\mathcal{I},\Pi,\Pi^{-1})$ be a family of qPRPs acting on bit strings of length $m+\tau$, where $m$ is the plaintext block length and $\tau$ is a randomness expansion parameter. Key generation samples a permutation key $k\leftarrow \mathcal{I}(1^{m+\tau})$. To encrypt a plaintext block $i\in \{0,1\}^m$, the client samples a fresh string $r\leftarrow\{0,1\}^\tau$ uniformly at random and computes
\[
Enc_k(i) = \pi_k (i||r).
\]
Decryption inverts the permutation and discards the appended randomness:
\[
Dec_k(j) := (\pi_k^{-1}(j))_m,
\]
where $(\cdot)_m$ denotes the first $m$ bits. Formal definition is given later in Sec~\ref{sec:security_proof}.

\subsection{Quantum One Time Padding}
A standard information-theoretic method for hiding an unknown quantum state is the quantum one-time pad (QOTP)~\cite{ambainis2000private, nielsen2010quantum}. For an \(n\)-qubit register, one samples uniformly random bit strings \(\mathbf{x},\mathbf{z}\xleftarrow{\$}\{0,1\}^n\) and applies the Pauli mask
\(X^{\mathbf x}Z^{\mathbf z}
=
\bigotimes_{j=1}^n X^{x_j}Z^{z_j}.\)
Here, the \(X^{\mathbf x}\) component flips computational-basis labels, while the \(Z^{\mathbf z}\) component adds basis-dependent phases. For a basis state \(\ket{i}\), these actions take the form
\[
X^{\mathbf x}\ket{i}=\ket{i\oplus x},
\qquad
Z^{\mathbf z}\ket{i}=(-1)^{z\cdot i}\ket{i}.
\]
Together, these two components constitute the standard Pauli \(X/Z\) masking layer.

The key property of QOTP is that averaging over the random key \((x,z)\) completely removes all information about the input state~\cite{ambainis2000private}. For any \(n\)-qubit density operator \(\rho\),
\[
\frac{1}{4^n}\sum_{x,z\in\{0,1\}^n}
X^{\mathbf x}Z^{\mathbf z}\,\rho\,Z^{\mathbf z}X^{\mathbf x}
=
\frac{I}{2^n}.
\]
Thus, to any observer who does not know \((x,z)\), the encrypted state is described by the maximally mixed state, independent of \(\rho\). In this sense, the plaintext is information-theoretically hidden: no measurement on the ciphertext can reveal any information about the original state beyond what is already contained in \(I/2^n\).

Equivalently, QOTP maps the ciphertext to the same average state \(I/2^n\) as a Haar-random pure-state ensemble. Thus, to any observer without the key, it reveals no information about the input state. This complete-randomization viewpoint is often useful in interpreting QOTP. QOTP and related Pauli masking techniques are widely used in delegated and blind quantum computation to conceal client-side quantum data while preserving coherent processability~\cite{broadbent2009universal,broadbent2015delegating,fitzsimons2017private}.

%% file: threatmodel.tex
\section{System and Threat Model}
\label{sec:threatmodel}

\begin{figure}[t]
         \centering
         \includegraphics[width=\linewidth]{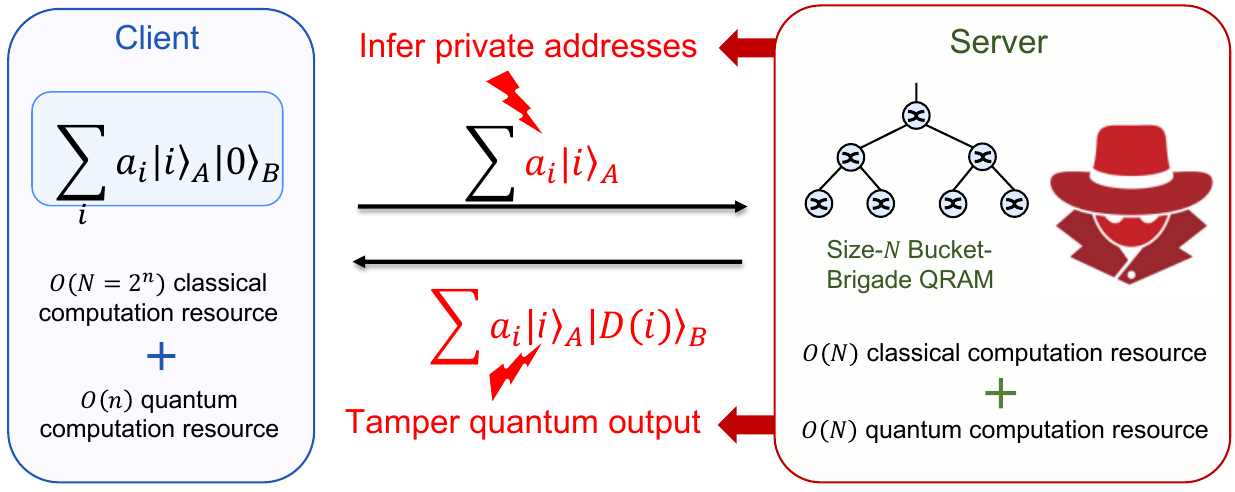}
         \caption{System model for secure quantum query. A trusted client outsources a protected database layout to an untrusted QRAM server, which stores the memory and performs quantum queries based on transmitted query address registers. The server may try to learn from the layout, query states, or interaction transcript; the goal is to hide the plaintext database and logical address state up to allowed public metadata.}
         \label{fig:threat_model}
\end{figure}

\subsection{System Model}

This work considers a client that initially owns a database \(D\), where each address \(i\in\{0,1\}^n\) indexes an \(m\)-bit record \(D[i]\in\{0,1\}^m\). The client outsources this database to an untrusted cloud server that provides a coherent QRAM interface for delegated memory access, namely
\[
\sum_i \alpha_i |i\rangle_A |0^m\rangle_B \longmapsto \sum_i \alpha_i |i\rangle_A |D[i]\rangle_B.
\]
The client does not upload \(D\) directly. Instead, before each refresh epoch, the client prepares a protected database consisting of encrypted records stored in a secret physical layout. The server stores this protected database and later serves coherent lookup queries over it.

Classical communication between the client and server is authenticated. Quantum registers are transmitted as part of the delegated QRAM interface; they are not hidden from the server, since the server necessarily acts on them. The protocol therefore protects the logical information of these quantum query states, rather than preventing the server from physically receiving quantum systems.

The security goal is to realize the above delegated coherent lookup functionality without revealing the client's logical address state \(\sum_i \alpha_i |i\rangle_A\) or the plaintext contents \(D\) of the outsourced database beyond public metadata such as database size and query timing.

\subsubsection{Client Capabilities}

The client is trusted. It owns the plaintext database at setup and refresh time, generates the secret cryptographic keys, samples local randomness, prepares the protected database layout, and synchronizes it with the server. During the online phase, the client performs only the local preprocessing and postprocessing required for each query. Its reliable quantum workspace consists of the address register, the data bus, and the ancillas used by the local protection circuit, for a total of \(O(n+m+anc)\) qubits, where \(anc\) denotes the protocol-implementation-dependent ancilla cost.
Classically, the client is assumed to be computationally efficient and capable of carrying out the key generation, permutation evaluation, and refresh procedures required by the protocol. Depending on the chosen refresh mechanism, the client may either rebuild the protected layout itself or rely on an additional secure reshuffling mechanism, as discussed in the deployment analysis.

\subsubsection{Server Capabilities}

The server is untrusted and provides the large memory and routing infrastructure needed to implement coherent QRAM over \(N=2^n\) logical addresses. The server stores the protected database but does not know the address-permutation key, the data-encryption key, the QOTP masks if used, the refresh randomness, or the locations of decoy queries.
For each online query, the server receives the protected query, coherently applies its QRAM lookup over the protected database, and returns the resulting registers to the client. The server has $O(N)$-scale storage and corresponding QRAM resources, including any ancillas required by its QRAM implementation.





\subsection{Adversarial Models}
Two adversarial models are considered.
\begin{itemize}
    \item \textbf{Honest-but-curious QRAM.} The server follows the prescribed offline storage interface and online QRAM execution, but attempts to infer information from everything it legitimately observes, including stored protected database, received online quantum query registers, refresh events, the quantum response registers before they are returned, and the classical communication transcript. The security goal of this model is privacy and confidentiality, which is targeted by the baseline oblivious QRAM protocol. The server should not learn the client’s logical address state or the plaintext database contents beyond allowed leakage.
    \item \textbf{Malicious QRAM.} The server is a quantum polynomial-time adversary that may deviate arbitrarily from the prescribed QRAM lookup. It may use polynomially many qubits of workspace to measure or partially measure query registers, entangle them with private workspace, perturb amplitudes or phases, return an incorrect ciphertext block, access the wrong physical cell, correlate behavior across rounds, or attempt to identify which queries are decoys. In this stronger model, the baseline protocol alone does not guarantee correct behavior. Malicious QRAM robustness is instead strengthened through decoy queries, which provide probabilistic detection of selected deviations such as invasive probing, inconsistent responses, coherence breaking, and output tampering.
\end{itemize}

\subsection{Leakage Profile}
The protocol does not attempt to hide all side information. The server may learn public parameters such as the database size \(N\), the number of queries, refresh timing \(t\), and whether an abort eventually occurs. If rejection is used for decoy failures, the client may delay or aggregate rejection decisions so that the timing of accept or reject behavior does not reveal decoy locations. Beyond this allowed leakage, the server should not learn the semantic meaning of logical addresses, the plaintext database, or the cross-epoch correspondence between logical and physical addresses, unless explicitly stated otherwise.

\subsection{Security Goals}
The protocol targets the following guarantees.
\begin{itemize}
    \item \textbf{Coherent Query Correctness.} For an honest server, the protected protocol should implement the same logical functionality as the ideal QRAM oracle \(O_D:\ket{i}_A\ket{b}_B \mapsto \ket{i}_A\ket{b\oplus D[i]}_B\), while preserving superpositions over \(i\).
    \item \textbf{Address privacy.} For an input query state \(\ket{\psi}_A=\sum_i \alpha_i\ket{i}_A\), the server should not learn the logical address labels or the amplitude/phase structure of \(\ket{\psi}_A\) beyond allowed leakage.
    \item \textbf{Data confidentiality.} The server stores encrypted database blocks and should learn no more about the plaintext contents than the allowed leakage profile.
    \item \textbf{Cheat sensitivity.} In the malicious model, active probing, wrong lookup, output tampering, or coherence-breaking behavior should be caught with nonzero probability by hidden tests. 
\end{itemize}

\subsection{Out-of-scope Guarantees}
This work targets privacy for delegated coherent QRAM access, not full blindness for an arbitrary delegated quantum computation. The server is allowed to know that it is providing a QRAM query service, the public database size, the refresh schedule, and the timing and number of query interactions, as specified in the leakage profile. What the protocol aims to hide is the plaintext database contents and the client's logical address state within each protected query. The protocol also does not by itself provide full verifiability of the server's quantum operation. In the malicious setting, decoy queries provide cheat sensitivity against certain active deviations, such as wrong lookup, output tampering, or coherence-breaking behavior, but they do not certify that the entire QRAM operation was implemented correctly. Physical leakage, timing channels beyond the stated leakage profile, denial-of-service attacks, and implementation-level attacks on the client's trusted device are outside the baseline model.


%% file: Protocol.tex
\section{Oblivious QRAM Protocol}
\label{sec:protocol}

\begin{figure*}[t]
         \centering
         \includegraphics[width=\linewidth]{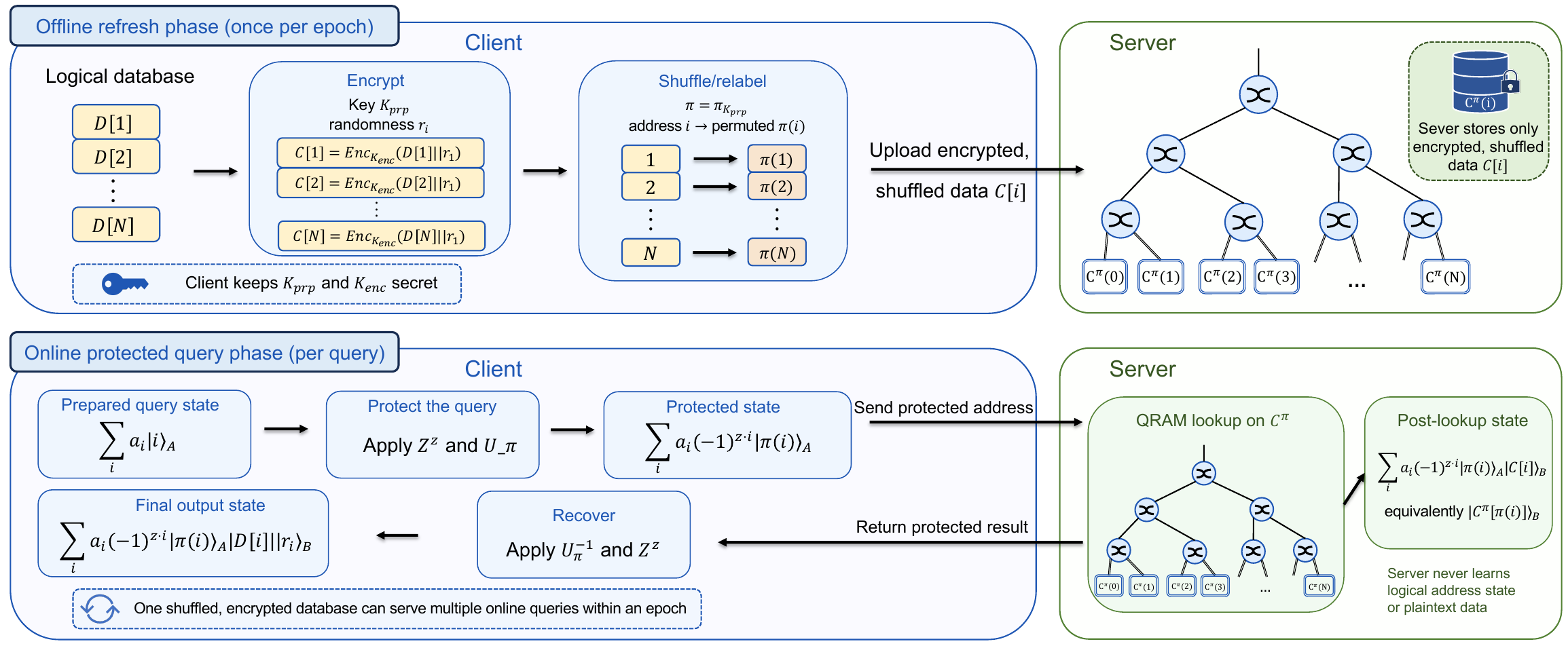}
         \caption{Overview of the qPRP-based OQRAM protocol. In the offline refresh phase, the client encrypts each database block, applies a secret qPRP-induced relabeling of logical addresses to physical QRAM locations, and uploads only the resulting protected layout \(C^\pi\) to the server. In the online phase, the client coherently maps an input address superposition into the protected physical namespace using \(Z^z\) and \(U_\pi\), the server performs a QRAM lookup over \(C^\pi\), and the client applies the inverse protection and decryption operations to recover the logical query output. Throughout the protocol, the server sees only protected addresses and encrypted, shuffled data.}
         \label{fig:fat}
\end{figure*}
This section presents the main oblivious QRAM protocol. Sec.~\ref{subsec:notation} first introduces the notation used throughout the construction. Sec.~\ref{subsec:overview} provides a high-level overview of the protocol. Sec.~\ref{subsec:offline-refresh} describes the offline refresh phase, where the client encrypts and shuffles the database. Sec.~\ref{subsec:online-query} then presents the online protected query phase, where the client masks the input query, invokes the server-side QRAM lookup, and coherently recovers the logical output. Sec.~\ref{subsec:two-round} discusses the extension to the general QRAM functionality with an arbitrary input bus register.

\subsection{Functionality and Notation}
\label{subsec:notation}

Let the address space client wants to query be \(\{0,1\}^n\) with \(N=2^n\), and let \(D:\{0,1\}^n\to\{0,1\}^m\) denote the database, where \(D[i]\) is the \(m\)-bit plaintext stored at original address \(i\).
Let \(\mathrm{Enc}_{K_{\mathrm{enc}}}:\{0,1\}^m\times\mathcal{R}\to\{0,1\}^{m_c}\) be a symmetric encryption scheme with key \(K_{\mathrm{enc}}\), $\tau$-bit randomness space \(\mathcal{R}\), and ciphertext length \(m_c = m+\tau\). For independently sampled randomness \(r_i\in\mathcal{R}\), define the encrypted database as \(C[i]:=\mathrm{Enc}_{K_{\mathrm{enc}}}(D[i];r_i)\). The corresponding query operation with the encrypted database is
\[
O_C:\ket{i}_A\ket{0^{m_c}}_B \mapsto \ket{i}_A\ket{C[i]}_B.
\]
Let \(\pi=\pi_{K_{\mathrm{prp}}}\) be a keyed permutation on \(\{0,1\}^n\), and let \(U_\pi\ket{i}=\ket{\pi(i)}\) denote its coherent action on the address register. Here, \(i\) denotes an original address, while \(j=\pi(i)\) denotes the corresponding physical address in the server-side permuted layout. The permuted encrypted layout is defined by \(C^\pi[j]:=C[\pi^{-1}(j)]\), equivalently \(C^\pi[\pi(i)]=C[i]\) for every \(i\in\{0,1\}^n\). The corresponding physical lookup oracle is
\[
O_{C^\pi}:\ket{j}_A\ket{0^{m_c}}_B \mapsto \ket{j}_A\ket{C^\pi[j]}_B.
\]

\subsection{Protocol Overview}
\label{subsec:overview}
At a high level, the Oblivious QRAM protocol proceeds in two phases, as summarized in Protocol~\ref{prot:oqram}. In the offline refresh phase, the client hides the dataset by reshuffling the database, encrypting the data blocks, and uploading only the resulting shuffled, protected database to the server. In the online phase, the client applies local quantum operations to map its address register into the same hidden physical namespace, allowing the server to answer the query without seeing the original logical addresses. The baseline construction realizes this hidden mapping through qPRP-based permutation and qIND-qCPA encryption, which supports multiple queries over a single uploaded database layout.

\begin{algorithm}[t]
\caption{\textsc{OQRAM}: qPRP-based protocol}
\label{prot:oqram}
\begin{algorithmic}[1]
\Statex \textbf{Offline shuffling phase (Shuffle \(D\)):}
\State Sample fresh \(K_{\mathrm{prp}}\), \(K_{\mathrm{enc}}\), and set \(\pi=\pi_{K_{\mathrm{prp}}}\).
\ForAll{\(i\in\{0,1\}^n\)}
    \State Sample \(r_i\leftarrow\{0,1\}^{\tau}\).
    \State Compute \(C[i]\gets \pi'_{K_{\mathrm{enc}}}(D[i]\|r_i)\).
    \State Set \(C^\pi[\pi(i)]\gets C[i]\).
\EndFor
\State Upload \(C^\pi\) to the server.

\Statex \textbf{Online protected-query phase (Query \(|\psi\rangle_A\)):}
\State Prepare \(|\psi\rangle_A=\sum_i\alpha_i|i\rangle_A\).
\State Sample \(z\leftarrow\{0,1\}^n\) and apply \(U_\pi Z^z\) to \(A\).
\State Send address register \(A\) to the server.
\State Server applies \(O_{C^\pi}:|j\rangle_A|0^{m+\tau}\rangle_B\mapsto |j\rangle_A|C^\pi[j]\rangle_B\).
\State Server returns \(A,B\) to the client.
\State Apply \(Z^zU_\pi^\dagger\) to recover the logical address basis.
\State Apply \((\pi'_{K_{\mathrm{enc}}})^{-1}\) to \(B\), obtaining \(\sum_i\alpha_i|i\rangle_A|D[i]\|r_i\rangle_B\).
\end{algorithmic}
\end{algorithm}

\subsection{Offline Refresh Phase}
\label{subsec:offline-refresh}
During the offline refresh phase, the client converts the logical database \(D=\{D[i]\}_{i\in\{0,1\}^n}\) into a protected one \(C^\pi\). An epoch denotes the interval during which the server uses one fixed protected database layout. In the qPRP-based baseline construction, each logical record is first encrypted with fresh block randomness, and the resulting ciphertexts are then placed according to a secret address permutation.

The qPRP-based baseline refresh procedure is as follows.

\begin{enumerate}
    \item The client samples a fresh qPRP key set \(K_{\mathrm{prp}}\), which defines the secret address permutation
    \[
        \pi=\pi_{K_{\mathrm{prp}}}:\{0,1\}^n\rightarrow\{0,1\}^n .
    \]
    \item The client samples a fresh encryption key set \(K_{\mathrm{enc}}\), determining qPRP-based qIND-qCPA encryption map
    \[
        \pi'_{K_{\mathrm{enc}}}:\{0,1\}^{m+\tau}
        \rightarrow
        \{0,1\}^{m+\tau}.
    \]
    \item For each logical address $i\in\{0,1\}^n$, the client samples independent fresh block randomness $r_i\xrightarrow{\$}\{0,1\}^\tau$.
    \item The client computes the ciphertext for each address
    \[C[i]=\mathrm{Enc}_{K_{\mathrm{enc}}}(D[i];r_i)=\pi'_{K_{\mathrm{enc}}}(D[i]||r_i).\]
    \item The client evaluates $\pi$ and constructs the protected address-data mapping by reshuffling the encrypted database to its permuted physical address:
    \[C^\pi[i]=C[\pi^{-1}(i)].\]
    Equivalently, \(C^\pi[\pi(i)]=C[i]\) for every logical address $i$.
    \item The client uploads the resulting protected database \(C^\pi\) to the server, and the server thereafter serves QRAM lookups over this fresh database throughout the current epoch.
\end{enumerate}

The resulting layout satisfies \(C^\pi[\pi(i)] = C[i]\) for every logical address \(i\). Thus, a QRAM lookup at physical address \(\pi(i)\) returns the correct ciphertext corresponding to logical address \(i\). This invariant is the link between the offline refresh phase and the online protected query phase. During the online query, the client coherently maps a logical address state into the hidden physical namespace defined by \(\pi\).

The two key sets have separate roles. \(K_{\mathrm{prp}}\) is the secret key for address-hiding qPRP, which induces a permutation on the $n$-bit logical address space and determines where each encrypted block is placed in physical storage. The key $K_{\mathrm{enc}}$ is the secret key for the qIND-qCPA secure encryption layer, which is realized through a qPRP-based map acting on $|D[i]|+\tau = m+\tau$ bits and is independent of the address permutation $\pi$. The per-block randomness $r_i$ ensures that repeated encryptions of the same plaintext yield independent ciphertexts. Thus, \(K_{\mathrm{prp}}\) determines the hidden physical locations, while \(K_{\mathrm{enc}}\) determines the ciphertext values.

When a new epoch begins, the client refreshes both the address permutation and the ciphertext representation. Refreshing only the permutation while leaving ciphertext values unchanged may allow the server to link cells across snapshots by equality of ciphertexts. Therefore, an epoch refresh re-randomizes both the physical placement and the encrypted block values.

A full refresh is not required after every query in the qPRP-based construction. Instead, one protected database layout may be reused for multiple online queries within a single epoch. The number of queries served per epoch is treated as a security parameter and is bounded by the multi-query security of the underlying qPRP and encryption scheme. This refresh interval trades offline communication and rebuild cost against stronger cross-query and cross-epoch unlinkability.

The protocol only requires the server to receive a database layout satisfying the above invariant. Different mechanisms for synchronizing \(C^\pi\) are discussed in Sec.~\ref{sec:results}.

\subsection{Online Masked Query}
\label{subsec:online-query}

This subsection describes the online query procedure executed over the protected database $C^\pi$ generated during the offline refresh phase. At a high level, this online phase enables the client to perform a logical QRAM lookup while ensuring that the server operates only on the protected representation of both the query and the database layout. The client coherently maps a logical address state into the hidden physical namespace defined by the secret permutation $\pi$, sends the protected query to the server. The server then evaluates QRAM access on the protected input and returns the resulting superposed ciphertexts, which the client decodes to recover the desired logical response. Throughout the online phase, the server interacts only with protected queries and data representations.

Let the client's logical address state be \(\ket{\psi}_A=\sum_i\alpha_i\ket{i}_A\), where $A$ is the $n$-qubit address register. The target logical functionality is
\[
\ket{\psi}_A \ket{0^m}_B \;\longmapsto\; \sum_i \alpha_i \ket{i}_A \ket{D[i]}_B.
\]
However, the protected database layout contains ciphertext blocks of length $m+\tau$, so the online ciphertext bus is an $m+\tau$-qubit register $B$. The server-side protected QRAM lookup is therefore modeled as the reversible oracle
\[
    U_{C^\pi}:\ket{i}_A\ket{b_i}_B
    \longmapsto
    \ket{i}_A\ket{b_i\oplus C^\pi[i]}_B .
\]
When $B$ is initialized to $\ket{0^{m+\tau}}$, this oracle loads the ciphertext stored at physical address $j$.

The qPRP-based online query proceeds as follows.

\begin{enumerate}
    \item The client prepares the original, logical address state $A$
    \[
    \ket{\psi}_A = \sum_i \alpha_i \ket{i}_A.
    \]
    
    \item The client additionally samples a fresh string \(z\in\{0,1\}^n\) and applies
    \(Z^z := \bigotimes_{k=1}^n Z^{z_k}\)
    before the qPRP mask. The client then applies the qPRP address permutation $\pi=\pi_{K_\mathrm{prp}}$ coherently to the address register, obtaining
    \[
    \ket{\widetilde{\psi}}_A = U_\pi Z^z \ket{\psi}_A
    = \sum_i \alpha_i (-1)^{z\cdot i} \ket{\pi(i)}_A.
    \]
    The client then sends the protected query $\ket{\widetilde{\psi}}_A$ to the server.

    \item The server executes the QRAM lookup over the protected database $C^\pi$ with an empty ciphertext bus state $ \ket{0^{m+\tau}}_B$
    \[
    U_{C^\pi}:\ket{j}_A\ket{0^{m+\tau}}_B \mapsto \ket{j}_A\ket{C^\pi[j]}_B.
    \]
    Thus, on the protected query, the server obtains the corresponding ciphertexts
    \[
    U_{C^\pi}\bigl(\ket{\widetilde{\psi}}_A\ket{0^{m+\tau}}_B\bigr)
    =
    \sum_i \alpha_i (-1)^{z\cdot i} \ket{\pi(i)}_A \ket{C^\pi[\pi(i)]}_B.
    \]
    
    \item The server returns the output to the client. Since $C^\pi[\pi(i)]=C[i]$ for every original address $i$, the resulting state is equivalently
    \[
    \sum_i \alpha_i (-1)^{z\cdot i} \ket{\pi(i)}_A \ket{C[i]}_B.
    \]

    \item The client applies \(Z^z U_\pi^\dagger\) to remove the address shuffling and phase mask, obtaining
\[
    \sum_i \alpha_i \ket{i}_A \ket{C[i]}_B .
\]

    \item The client applies the inverse encryption permutation \((\pi'_{K_{\mathrm{enc}}})^{-1}\) coherently to the ciphertext bus
    \((\pi'_{K_{\mathrm{enc}}})^{-1}\ket{C[i]}_B
        =
        \ket{D[i]\|r_i}_B .
    \)
        
    The one-round protected query therefore produces
    \[
        \sum_i \alpha_i \ket{i}_A\ket{D[i]\|r_i}_B .
    \]
\end{enumerate}

The final bus contains both the \(m\)-qubit plaintext block and the \(\tau\)-qubit randomness suffix. Although the values \(r_i\) are sampled independently of the plaintext records, once the protected layout is fixed they become branch-dependent labels correlated with the address \(i\). Therefore, the randomness suffix cannot simply be traced out if a coherent QRAM output is required. Indeed, for
\[
    \ket{\Psi}_{ADR}
    =
    \sum_i \alpha_i \ket{i}_A\ket{D[i]}_D\ket{r_i}_R ,
\]
tracing out \(R\) gives
\[
    \rho_{AD}
    =
    \sum_{i,j}
    \alpha_i\alpha_j^*
    \langle r_j|r_i\rangle
    \ket{i,D[i]}\bra{j,D[j]} .
\]
Thus, coherence between branches \(i\) and \(j\) in the reduced \(AD\) state is preserved only when \(r_i=r_j\). Since independently sampled \(r_i\)'s and \(r_j\)'s are distinct with high probability for large \(\tau\), discarding the randomness register can dephase the address-data superposition.

Consequently, the one-round online query should be interpreted as coherent recovery of the expanded decrypted block \(\ket{D[i]\|r_i}\), rather than as a clean realization of the ideal plaintext oracle alone. The randomness suffix does not by itself destroy coherence as long as it remains part of the coherent quantum state and is never measured, discarded, reset, or otherwise coupled to the environment. In this sense, the client may still use the plaintext portion coherently, for example by applying reversible computation controlled on, or acting on, the first \(m\) qubits of \(B\), while leaving the branch-dependent suffix as untouched garbage until the end of the computation.

However, algorithms that require a clean logical query oracle (e.g., recycling the unused qubits) cannot in general ignore this garbage. A clean realization of
\(
    O_D:\ket{i}_A\ket{b_i}_M
    \longmapsto
    \ket{i}_A\ket{b_i\oplus D[i]}_M
\)
requires extracting the \(m\)-qubit plaintext into a separate algorithmic bus and then uncomputing the ciphertext/randomness workspace. This query-use-unquery extension is described next in Sec.~\ref{subsec:two-round}.

\subsection{Two-round Protocol for Arbitrary Bus State by Uncomputing the Bus Register}
\label{subsec:two-round}

The one-round protocols above were presented in the standard retrieval form, where the query bus is initialized to a clean zero state before the protected QRAM access. In some applications, however, the client queries QRAM as part of a larger coherent computation and therefore requires a reversible oracle that acts on an existing algorithmic bus register. The desired logical functionality is
\[
O_D:\ket{i}_A\ket{b_i}_{M} \mapsto \ket{i}_A\ket{b_i\oplus D[i]}_{M},
\] 
while restoring any temporary QRAM workspace to $\ket{0^{m+\tau}}_B$. Here $A$ is the $n$-qubit address register and $M$ is the $m$-qubit algorithmic bus register kept locally by the client for its own computation, whose state is independent of the protected QRAM procedure.

In the protected setting, the server stores an encrypted protected layout $C^\star$, where $C^\star$ denotes either $C^\pi$ in the qPRP-based variant or $C^\mathbf{x}$ in the qOTP-based variant introduced later in Sec.~\ref{subsec:qotp-based}. Let $\pi^\star$ denote the corresponding protected address map. The server-side QRAM oracle acts on a protected address register $A$ and a temporary $(m+\tau)$-qubit ciphertext bus $B$ in reversible XOR-loading form:
\[
O_{C^\star}:\ket{j}_A\ket{b_j}_B \longmapsto \ket{j}_A\ket{b_j\oplus C^\star[j]}_B .
\]

Because this oracle satisfies \(O_{C^\star}^2=I\), the same protected QRAM access can be used once to load the ciphertext block into \(B\) and a second time to erase it. Between these two calls, the client removes the address protection, decrypts the expanded block \(D[i]\|r_i\), CNOTs the plaintext portion into the local algorithmic bus \(M\), re-encrypts the temporary bus, and reapplies the address protection. This is the standard query-use-unquery pattern applied to the protected QRAM oracle, rather than a fundamentally different primitive.

Consequently, starting from
\(
\sum_i \alpha_i \ket{i}_A \ket{0^{m+\tau}}_B \ket{b_i}_{M}
\),
the two-round wrapper implements
\[
\sum_i \alpha_i \ket{i}_A \ket{0^{m+\tau}}_B \ket{b_i}_{M}
\;\longmapsto\;
\sum_i \alpha_i \ket{i}_A \ket{0^{m+\tau}}_B \ket{b_i\oplus D[i]}_{M},
\]
with the ciphertext bus used only as temporary workspace. The detailed step-by-step state evolution is given in Appendix~\ref{subsec:two-round-detailed}.

%% file: variant.tex
\section{Protocol Variants and Extensions}
\label{sec:variants}
The baseline qPRP-based construction supports protected QRAM queries across a refresh epoch. This section presents two extensions: a qOTP-based single-query masking variant with lower online quantum cost but per-query refresh in Sec.~\ref{subsec:qotp-based}, and decoy queries that provide cheat-sensitive detection of malicious server deviations in Sec.~\ref{subsec:decoys}.

\subsection{QOTP-Based Single-Query Masking}
\label{subsec:qotp-based}
The qOTP-based variant is the single-query analogue of the qPRP baseline. It keeps the same encrypted-block layer, server-side QRAM lookup, and client-side ciphertext recovery, but replaces the coherent qPRP address permutation with a fresh quantum one-time pad on the address register~\cite{ambainis2000private,nielsen2010quantum}:
\[
X^{\mathbf{x}}Z^{\mathbf{z}}
=
\bigotimes_{k=1}^n X^{x_k}Z^{z_k},
\qquad
\mathbf{x},\mathbf{z}\xleftarrow{\$}\{0,1\}^n .
\]
The main structural difference is the \(X\)-mask. It shifts basis addresses by \(i\mapsto i\oplus\mathbf{x}\), so the offline layout is shifted rather than permuted:
\[
C^{\mathbf{x}}[j]=C[j\oplus \mathbf{x}],
\qquad
C^{\mathbf{x}}[i\oplus \mathbf{x}]=C[i].
\]
The \(Z\)-mask contributes only phases to the address superposition and does not affect classical storage locations.

Online query phase matches the qPRP online phase, replacing \(U_\pi Z^z\) with \(X^{\mathbf{x}}Z^{\mathbf{z}}\) and querying the shifted layout \(C^{\mathbf{x}}\). For an input state \(\sum_i\alpha_i\ket{i}_A\), the server receives
\(
\sum_i \alpha_i (-1)^{\mathbf{z}\cdot i}\ket{i\oplus\mathbf{x}}_A
\),
loads the corresponding ciphertext \(C^{\mathbf{x}}[i\oplus\mathbf{x}]=C[i]\), and returns the protected address-ciphertext state. The client then removes the Pauli mask, decrypts \(C[i]\) coherently, and proceeds exactly as in the qPRP-based protocol.

Unlike the qPRP baseline, this masking does not give a reusable pseudorandom relabeling of the full classical address space. The induced map is only an XOR shift, so reusing the same \((\mathbf{x},\mathbf{z})\) and shifted layout across independent queries can allow cross-round correlations under a fixed hidden shift. The client therefore samples fresh \((\mathbf{x},\mathbf{z})\), re-encrypts the database blocks, and rebuilds the shifted layout after each independent query.

The detailed offline and online steps are given in Appendix~\ref{subsec:qotp-detailed}. The online quantum cost of the qOTP masking and unmasking is $O(n)$ single-qubit Pauli gates; the main cost is the per-query classical refresh.

\subsection{Decoy Queries for Malicious Server Detection}
\label{subsec:decoys}

The offline and online phases defined above protect query privacy and data confidentiality against an honest-but-curious server, but by themselves do not prevent a malicious server from deviating from the prescribed lookup or returning an incorrect output. To obtain a probabilistic honesty check, the client can interleave the protected query process with decoy queries~\cite{giovannetti2008privatequery}. Each decoy query is executed through exactly the same offline and online protection procedures as an ordinary one. The only difference is that the client chooses certain query instances for which the returned output can be checked. The decoy choice is made locally by the client and is not announced before the server returns the query result. Therefore, from the server’s perspective, a protected query should not reveal whether it is an ordinary request or a verification round.

For each query instance, the client samples a hidden decoy bit \(b_{decoy}\in\{0,1\}\) with
\(\Pr[b_{decoy}=1]=p_{\mathsf{decoy}}\).
If \(b_{decoy}=0\), the client prepares the intended logical query state. If \(b_{decoy}=1\), the client prepares a decoy query that is protected using the same oblivious QRAM procedures.

A simple decoy is a known-answer query. If the client knows the classical value \(D[i_{decoy}]\) for a selected basis address \(i_{decoy}\), then after recovery the client can measure the returned plaintext block and compare it with the expected value. 
This check can directly verify correctness of the returned value for the selected address with certainty, assuming an ideal measurement and a classical database block. However, it is limited to single-address queries and does not test the server’s action on an arbitrary address superposition.

The main quantum decoy considered here is an inversion check on the protected QRAM oracle that can be applied to arbitrary protected address state, including a superposition.

For a protected decoy address $\ket{\widetilde{\psi}_{\mathsf{decoy}}}_A = \sum_i \alpha_i (-1)^{z\cdot i}\ket{\pi^\star(i)}_A$, the first QRAM call gives $\sum_i \alpha_i (-1)^{z\cdot i}\ket{\pi^\star(i)}_A\ket{C^\star[\pi^\star(i)]}_B$.
The client keeps the returned ciphertext bus without decrypting it. After a delay, possibly with ordinary queries interleaved, the client sends the same protected address state together with the stored bus back to the server. If the server applies the same QRAM oracle consistently, the second call erases the bus
\[
O_{C^\star}(\sum_i \alpha_i (-1)^{z\cdot i}\ket{\pi^\star(i)}_A\ket{C^\star[\pi^\star(i)]}_B) =
\ket{\widetilde{\psi}_{\mathsf{decoy}}}_A\ket{0^{m+\tau}}_B.
\]
Equivalently, the check test $O_{C^\star}^2 = I$ on the selected protected state. The client accepts only if measuring the returned bus gives the all-zero outcome.

This test detects deviations that prevent the second call from uncomputing the bus, but does not certify plaintext correctness, since a consistently incorrect lookup may still pass. It can therefore be combined with known-answer decoys. The two calls must use the same protected database layout $C^\star$.
In the qOTP-based variant, they also use the same address mask and shifted layout. This does not violate the per-query refresh requirement because the two calls form one query–unquery verification pair on the same address superposition, not two independent queries. A fresh qOTP mask and shifted layout are still required before any independent query.

Any failed decoy check causes rejection. To avoid revealing which rounds were decoys, the client may delay announcing the comparison outcome or defer rejection until after a batch of queries. This prevents the server from immediately learning the decoy locations from the timing of rejection. The detection bound is analyzed in detail in Sec.~\ref{subsec:decoy_bound}.

%% file: proof.tex
\section{Security Discussion}
\label{sec:security_proof}
This section states the security assumptions used by the OQRAM construction in Sec.~\ref{sec:protocol} and then analyzes the honest-but-curious privacy guarantee in Sec.~\ref{subsec:hbc-privacy}. The proof separates address privacy from database confidentiality. Address privacy follows from phase padding, secret relabeling, and the strong quantum-security of the qPRP, while database confidentiality follows from the qIND-qCPA security of the encryption layer. Malicious server detection is handled separately by the decoy mechanism in Sec.~\ref{subsec:decoys}.

\subsection{Correctness}
As already explained in Sec.~\ref{subsec:online-query}, correctness follows the database layout invariant
\[
    C^{\star}[\pi^\star(i)]
    =
    C[i]
    =
    Enc_{K_{\mathrm{enc}}}(D[i];r_i)
    =
    \pi'_{K_{\mathrm{enc}}}(D[i]\|r_i).
\]
That is, after the client applies the coherent address hiding $U_{\pi^\star}$, a QRAM lookup at address $\pi^\star(i)$ returns the ciphertext associated with the original logical address $i$. The client then applies $U_{\pi^\star}^{-1}$ to restore the logical address and coherently decrypts the returned ciphertext by applying \((\pi'_{K_{\mathrm{enc}}})^{-1}\). Therefore, for any query superposition \(\sum_i \alpha_i |i\rangle_A |0\rangle_B\), an honest execution implements
\[
    \sum_i \alpha_i |i\rangle_A |0\rangle_B
    \mapsto
    \sum_i \alpha_i |i\rangle_A |D[i]\rangle_B
\]
up to the auxiliary encryption randomness register, which can be optionally uncomputed as described in Sec.~\ref{subsec:two-round}.

\subsection{Cryptographic Assumptions}
\label{subsec:crypto-assumptions}
The constructions of the address permutation and encryption layer were introduced in the background section. This section then only states the security guarantees used later in the proof.

\begin{definition}[Strong quantum-secure PRP]
\label{def:sqprp}
Let \(\Pi=\{\pi_K\}_K\) be a keyed permutation family over \(\{0,1\}^{\ell}\). For a quantum adversary \(\mathcal{A}\) with quantum oracle access to both the forward and inverse maps, define
\[
\operatorname{Adv}^{\mathsf{sqPRP}}_{\Pi}(\mathcal{A})
=
\left|
\Pr_K[\mathcal{A}^{\pi_K,\pi_K^{-1}}(1^\lambda)=1]
-
\Pr_P[\mathcal{A}^{P,P^{-1}}(1^\lambda)=1]
\right|,
\]
where \(P\) is a uniformly random permutation over \(\{0,1\}^{\ell}\), and $\lambda$ is the security parameter. The family \(\Pi\) is a strong qPRP if this advantage is negligible for every QPT adversary making polynomially many bidirectional quantum queries.
\end{definition}

\begin{theorem}[Seven-round Feistel qPRP]
\label{thm:feistel-sqprp}
The seven-round balanced Feistel construction instantiated with quantum-secure round functions is a strong qPRP~\cite{zhandry2025note}.
Equivalently, for every adversary making \(q\) forward/inverse quantum queries,
\[
\operatorname{Adv}^{\mathsf{sqPRP}}_{\Pi^{(7)}}(q,\lambda)
\leq
\varepsilon_{\mathsf{prp}}(q,\lambda),
\]
where \(\varepsilon_{\mathsf{prp}}\) is the distinguishing bound from the seven-round Feistel theorem.
\end{theorem}

The protocol uses this qPRP guarantee for both the address permutation \(\pi_{K_{\mathrm{prp}}}\) over $n$-bit address and the encryption permutation \(\pi'_{K_{\mathrm{enc}}}\) over ($m+\tau$)-bit data blocks.

\begin{definition}[qIND-qCPA security]
\label{def:qind-qcpa}
A symmetric-key encryption scheme is qIND-qCPA secure~\cite{boneh2013quantum} if every QPT adversary with quantum chosen-plaintext access wins the qIND challenge game with probability at most
\(
    \frac12+\operatorname{negl}(\lambda)
\).
Equivalently, encryptions of equal-length challenge plaintext states are computationally indistinguishable to any such adversary.
\end{definition}

\begin{theorem}[qPRP-based qIND-qCPA encryption]
\label{thm:qprp-qind}
Let \(\Pi_{m+\tau}\) be a qPRP over \(\{0,1\}^{m+\tau}\). The encryption scheme
\[
    \mathsf{Enc}_{K_{\mathrm{enc}}}(D;r)
    =
    \pi'_{K_{\mathrm{enc}}}(D\|r),
    \qquad
    r \xleftarrow{\$}\{0,1\}^{\tau},
\]
with decryption given by applying \((\pi'_{K_{\mathrm{enc}}})^{-1}\) and outputting the first \(m\) bits, is qIND-qCPA secure~\cite{boneh2013quantum}.
Thus, for every QPT adversary \(\mathcal{A}\),
\[
\operatorname{Adv}^{\mathsf{qIND\text{-}qCPA}}_{\mathsf{Enc}}(\mathcal{A})
\leq
\varepsilon_{\mathsf{enc}}(\lambda),
\]
for negligible \(\varepsilon_{\mathsf{enc}}\).
\end{theorem}

Therefore, the proof may replace the real address permutation and encryption layer by ideal random objects with additive loss bounded by \(    \varepsilon_{\mathsf{addr}}(q_{\mathsf{addr}},\lambda) + N\cdot \varepsilon_{\mathsf{enc}}(\lambda)\), over an epoch containing one lookup and an $N=2^n$-block encrypted database layout.

\subsection{Honest-but-Curious Privacy}
\label{subsec:hbc-privacy}

\subsubsection{qPRP-based}

The honest-but-curious server follows the specified QRAM lookup but may retain all classical and quantum side information it receives. In this setting, the privacy guarantee has three layers. First, fresh \(Z\)-padding removes phase information from the transmitted address state. Second, if each epoch serves only a single query, namely every query uses a fresh permutation and the encrypted database layout is refreshed consistently, then the server's one-query view is information-theoretically independent of the client's input address state. Third, if the same hidden permutation of the database is reused across multiple queries in one epoch, then the guarantee becomes weaker: the server still does not learn which original addresses the observed physical labels correspond to, but it may learn relabeling-invariant statistics of the query ensemble, such as repeated support, overlap patterns, or diagonal weight distributions in the hidden basis.

Starting with the effect of \(Z\)-padding alone.

\begin{lemma}[\(Z\)-padding dephases the address register]
\label{lem:z-dephase}
For any \(n\)-qubit address state \(\rho\),
\[
\mathbb E_z[Z^z\rho Z^z]=\Delta(\rho),
\]
where \(\Delta(\rho)=\sum_x \rho_{x,x}|x\rangle\langle x|\).
\end{lemma}

\begin{proof}
Write \(\rho=\sum_{x,y}\rho_{x,y}|x\rangle\langle y|\). For any \(z\),
\(
Z^z|x\rangle\langle y|Z^z=(-1)^{z\cdot(x\oplus y)}|x\rangle\langle y|
\).
If \(x=y\), the phase is always \(1\). If \(x\neq y\), averaging over uniform
\(z\) gives zero. Therefore, all off-diagonal terms vanish and the diagonal terms
remain unchanged.
\end{proof}

If the protocol uses a fresh independent permutation and fresh \(Z\)-padding for every query, and the database permutation and encryption are refreshed consistently, then the server's one-query view is maximally mixed.

\begin{lemma}[Single-query rerandomization]
\label{lem:single-query-rerandomization}
Let \(\pi\) be a uniformly random permutation on \(\{0,1\}^n\), let
\(z\leftarrow\{0,1\}^n\) be uniform, and let \(N=2^n\). For any address state
\(\rho_A\),
\[
\mathbb E_{\pi,z}\left[U_\pi Z^z\rho_A Z^zU_\pi^\dagger\right]=\frac{I_A}{N}.
\]
Moreover, for any joint state \(\rho_{AE}\),
\[
\mathbb E_{\pi,z}
\left[
(U_\pi Z^z\otimes I_E)\rho_{AE}(Z^zU_\pi^\dagger\otimes I_E)
\right]
=
\frac{I_A}{N}\otimes \rho_E .
\]
\end{lemma}

\begin{proof}
The first identity follows by applying Lemma~\ref{lem:z-dephase} and then
averaging the resulting diagonal distribution over a uniformly random
permutation. Every basis state receives total weight \(1/N\). For the statement
with side information, expand
\(\rho_{AE}=\sum_{x,y}|x\rangle\langle y|\otimes R_{x,y}\). The \(Z\)-average
removes all \(x\neq y\) blocks, and the permutation average maps
\(\sum_x |x\rangle\langle x|\otimes R_{x,x}\) to
\(\frac{I_A}{N}\otimes \sum_x R_{x,x}\), which is
\(\frac{I_A}{N}\otimes\rho_E\).
\end{proof}

However, the claim above doesn't apply to qOTP-variant.
\begin{lemma}[Vulnerability of reused one-time padding]
\label{lem:qotp-refresh-vulnerability}
Let \(x\in\{0,1\}^n\) be fixed across \(T\) queries, while each
\(z_t\leftarrow\{0,1\}^n\) is sampled independently. For address states
\(\rho_A^{(1)},\ldots,\rho_A^{(T)}\), define
\[
\mathcal E_x^{(T)}
\left(\rho_A^{(1)},\ldots,\rho_A^{(T)}\right)
=
\mathbb E_{z_1,\ldots,z_T}
\bigotimes_{t=1}^T
X^x Z^{z_t}\rho_A^{(t)} Z^{z_t}X^x .
\]
Then
\[
\mathcal E_x^{(T)}
\left(\rho_A^{(1)},\ldots,\rho_A^{(T)}\right)
=
\bigotimes_{t=1}^T
X^x \Delta(\rho_A^{(t)})X^x ,
\]
where
\(\Delta(\rho)=\sum_i |i\rangle\langle i|\rho|i\rangle\langle i|\)
is dephasing in the computational basis. Hence fresh \(Z\)-masks remove
phase coherence, but a reused \(X\)-mask preserves a fixed relabeling of
the diagonal address distributions across queries.
\end{lemma}

\begin{proof}
For each query \(t\), averaging over \(z_t\) removes all off-diagonal
blocks in the computational basis:
\(
\mathbb E_{z_t}\left[Z^{z_t}\rho_A^{(t)}Z^{z_t}\right]
=
\Delta(\rho_A^{(t)})
\).
The fixed \(X^x\) mask then only relabels basis states by
\(i\mapsto i\oplus x\), giving the stated expression. This is not equivalent
to independent single-query rerandomization. For example, if \(x\) is sampled
once and then reused for two queries with
\(\rho_A^{(1)}=\rho_A^{(2)}=|0\rangle\langle 0|\), the server's averaged
two-query state is
\(
\frac{1}{N}\sum_x |x\rangle\langle x|\otimes |x\rangle\langle x|
\),
so measuring both protected addresses gives the same outcome with
probability \(1\). Under two independently rerandomized queries, the state
would be \(I_A/N\otimes I_A/N\), and the same test succeeds with probability
only \(1/N\). Thus reusing the \(X\)-mask leaks cross-query correlations,
which is why the qOTP-based layout must be refreshed after each query.
\end{proof}

\begin{lemma}[Computational address hiding]
\label{lem:computational-address-hiding}
If \(\pi_{K_{\mathrm{prp}}}\) is a strong qPRP, then the real address-side view using \(\pi_{K_{\mathrm{prp}}}\) is indistinguishable from the view using a uniformly random permutation, except with advantage
\(
    \varepsilon_{\mathsf{addr}}(q_{\mathsf{addr}},\lambda)
\).
\end{lemma}

\begin{proof}
Any server distinguishing these two experiments gives a qPRP distinguisher by using its oracle to generate the protected address registers.
\end{proof}

\begin{lemma}[Database confidentiality]
    \label{lem:database-confidentiality}
    If \(Enc_{K_{\mathrm{enc}}}\) is qIND-qCPA secure, then the encrypted layout \(C^\pi\) is indistinguishable from an encrypted dummy layout \(D^\star[i]=0^m\), except with advantage at most
    \(
    N\cdot\varepsilon_{\mathsf{enc}}(\lambda).
    \)
\end{lemma}
\begin{proof}
    Use hybrids \(H_0,\ldots,H_N\), where $H_l$ encrypts the first $l$ blocks from the dummy database and the remaining blocks from the real database. Adjacent hybrids differ in one ciphertext and are indistinguishable by qIND-qCPA security. Summing over $N$ hybrids gives the bound.
\end{proof}

\begin{theorem}[Honest-but-curious qPRP-based single query privacy]
  \label{thm:hbc-prp-single-query}
  Assuming strong qPRP security for \(\pi_{K_{\mathrm{prp}}}\) and qIND-qCPA
security for \(Enc_{K_{\mathrm{enc}}}\). In a single-query epoch with fresh keys, fresh encryption randomness, and fresh $Z$-padding, the server's view is simulatable from the ideal leakage
\[
    \mathcal{L_{qPRP}}
    =
    \left(
        N,\,
        m+\tau,\,
        1,\,
       \{|\alpha_x|^2:x\in\{0,1\}^n\}
    \right),
\]
where \(N\) is the number of QRAM cells, \(m+\tau\) is the ciphertext block length, \(1\) denotes that one lookup occurred, and \(\{|\alpha_x|^2:x\in\{0,1\}^n\}\) is the relabeling-invariant dephased query distribution. The distinguishing advantage is bounded by
\(
    \varepsilon_{\mathsf{addr}}(q_{\mathsf{addr}},\lambda)
    +
    N\cdot\varepsilon_{\mathsf{enc}}(\lambda)
\).
\end{theorem}
\begin{proof}[Proof sketch]
    Replace the address qPRP by a uniformly random permutation using Lemma~\ref{lem:computational-address-hiding}. With fresh $Z$-padding, the address register is dephased. The fresh random permutation hides the association between logical labels and physical labels, leaving only the unlabeled diagonal distribution \(\{|\alpha_i|^2\}\). Then replace the encrypted database by an encrypted dummy database using Lemma~\ref{lem:database-confidentiality}. The resulting view depends only on \(\mathcal{L}\), and the total loss is the sum of the address-hiding and database-confidentiality hybrid losses. A detailed proof is in Appendix~\ref{subsec:hbc-qprp-proof}.
\end{proof}

Although \(\mathcal{L}_{\mathsf{qPRP}}\) includes the dephased amplitudes \(\{|\alpha_i|^2\}\), learning this distribution requires many repeated copies of the same protected input. After $Z$-padding, the server only sees a shuffled diagonal distribution, whose full reconstruction over $N$ addresses requires roughly \(\tilde{\Omega}(N)\) copies in standard tomography~\cite{haah2016sample,o2016efficient}. Therefore, refreshing after \(t<N^{1/12}\) queries keeps this leakage far below the reconstruction regime.

If instead one hidden permutation is reused across multiple queries in one epoch, then the privacy guarantee is weaker. The server may correlate the protected queries it receives across rounds, but any successful attempt to exploit the specific structure of the keyed permutation family to recover the hidden permutation would imply a distinguisher against the underlying qPRP.

\begin{lemma}[Fixed-epoch multi-query hiding under qPRP]
\label{lem:qprp-multiquery}
Fix one epoch and let \(\pi=\pi_{K_{\mathrm{prp}}}\) be the hidden permutation determined by a secret qPRP key \(K_{\mathrm{prp}}\). Suppose the server stores the corresponding encrypted layout \(C^\pi\), and across the epoch receives \(t\) protected queries encoded using the same \(\pi\), with fresh independent \(Z\)-padding in each round. If the qPRP family is secure against \(q\) quantum uses and \(t<q\), then for any efficient quantum server, the resulting multi-query view is computationally indistinguishable from the same experiment in which \(\pi\) is replaced by a uniformly random hidden permutation. In particular, the server cannot recover the original-to-physical address mapping, or learn the qPRP key itself, beyond what is possible in the random hidden-permutation experiment.
\end{lemma}

\begin{proof}
The proof is by a standard reduction to qPRP security. If an efficient server could distinguish the real experiment from the random hidden-permutation experiment with non-negligible advantage after observing \(t<q\) protected queries, then one could build a quantum distinguisher that uses the server as a subroutine and breaks the qPRP family with the same non-negligible advantage.
\end{proof}

\subsubsection{qOTP-based}

\begin{lemma}[qOTP address hiding]
\label{lem:qotp-address-hiding}
For any \(n\)-qubit address state \(\rho_A\),
\[
    \mathbb E_{a,b}\left[X^a Z^b \rho_A Z^b X^a\right]
    =
    \frac{I_A}{N}, \qquad
    N=2^n.
\]
More generally, for any joint state \(\rho_{AE}\),
\[
    \mathbb E_{a,b}\left[(X^aZ^b\otimes I_E)\rho_{AE}(Z^bX^a\otimes I_E)\right]
    =
    \frac{I_A}{N}\otimes \rho_E.
\]
\end{lemma}
\begin{proof}
Write \(\rho_A=\sum_{i,j}\rho_{i,j}|i\rangle\langle j|\).
First of all, averaging over \(b\) can remove all off-diagonal terms: 
\(\mathbb E_b[Z^b\rho_A Z^b] = \sum_i \rho_{i,i}|i\rangle\langle i|.\)
Then, averaging over \(a\) uniformly shifts the diagonal distribution: 
\(\mathbb E_a\left[X^a\left(\sum_i \rho_{i,i}|i\rangle\langle i|\right)X^a\right]=\frac{I_A}{N}.\)
The proof with side information is identical after writing \(\rho_{AE}=\sum_{i,j}|i\rangle\langle j|\otimes R_{i,j}\).
The \(Z\)-average removes the \(i\neq j\) blocks, and the \(X\)-average maps the
remaining diagonal address register to \(I_A/N\), while preserving \(\rho_E=\sum_i R_{i,i}\).
\end{proof}

\begin{theorem}[Honest-but-curious qOTP-based single query privacy]
\label{thm:qotp-hbc}
    Assume that \(Enc_{K_{\mathrm{enc}}}\) is qIND-qCPA secure. In a single-query epoch using fresh address pads \(a,b\), fresh encryption randomness, and a layout refreshed consistently with \(a\), the honest-but-curious server's view is simulatable from the leakage
    \[
        \mathcal{L}_{\mathsf{qOTP}}=\left(N,\,m+\tau,\,1\right),
    \]
    where \(N\) is the number of cells, \(m+\tau\) is the ciphertext block length, and \(1\) denotes that one lookup occurred. The distinguishing advantage is bounded by
    \(N\cdot \varepsilon_{\mathsf{enc}}(\lambda)\).
\end{theorem}
\begin{proof}[Proof sketch]
For any query state \(\rho_A\), the fresh qOTP satisfies \(\mathbb E_{a,b}[X^aZ^b\rho_A Z^bX^a]=\frac{I_A}{N}\). Thus the protected address register is information-theoretically independent of the logical query. The only remaining information is the refreshed encrypted layout. By qIND-qCPA security and an \(N\)-block hybrid, this layout is indistinguishable from an encrypted dummy layout with loss \(N\varepsilon_{\mathsf{enc}}(\lambda)\). The resulting view consists only of a maximally mixed address register, an encrypted dummy layout, and the public parameters \((N,m+\tau,1)\), so it is simulatable from \(\mathcal{L}_{\mathsf{qOTP}}\). A detailed proof is in Appendix~\ref{subsec:hbc-qotp-proof}.
\end{proof}

\subsection{Security of the Two-Round Extension}
\label{subsec:two-round-security}

The two-round extension is a query-use-unquery wrapper around the one-round protected query, and it does not introduce a new hiding mechanism.

\begin{lemma}[Two-round privacy preservation]
\label{lem:two-round-privacy-preservation}
If the one-round protected query hides the logical address state and
plaintext database from an honest-but-curious server up to the allowed
leakage, then the two-round extension preserves the same privacy
guarantee. One logical query is counted as two protected oracle
invocations.
\end{lemma}

\begin{proof}
Let \(L\) denote the private logical information, and let \(V_1,V_2\) be
the server's views in the first and second protected lookups. The only
operation between the two lookups that touches plaintext is performed
locally by the client: the register containing \(b_i\oplus D[i]\) is never
transmitted. Before the second lookup, the transmitted registers are 
in the protected address-ciphertext seen by the server.

Therefore, conditioned on the first view and the public fact that the two
calls form a query-use-unquery pair, the second view is simulatable from the first:
\[
    I(L;V_2\mid V_1,\mathsf{pub})=0 .
\]
The server only learns that the second call uncomputes the first, which
reveals neither the logical address nor the plaintext data.

For the qPRP variant, the wrapper only doubles the protected-oracle count
within the epoch security budget. For the qOTP variant, the same mask and
shifted layout are reused only inside this query-use-unquery pair; every
independent logical query still requires a fresh mask and refreshed layout.
The claim is limited to honest-but-curious privacy and does not provide
malicious-server verifiability.
\end{proof}

\subsection{Decoy Detection Bound}
\label{subsec:decoy_bound}

\begin{proposition}[Decoy detection bound]
\label{prop:decoy-bound}
Fix an attack family \(\mathcal A\). Suppose each attacked round is a decoy with probability \(p\), independently of the server's view, and that conditioned on an attacked round being a decoy, the client rejects with probability at least \(\eta\). Then the probability that \(T\) attacked rounds all escape detection is at most
\[
(1-p\eta)^T\le e^{-p\eta T}.
\]
\end{proposition}

\begin{proof}
For each attacked round, the probability of escaping detection is at most \(1-p\eta\). Multiplying over \(T\) rounds gives \((1-p\eta)^T\), and the exponential bound follows from \(1-x\le e^{-x}\).
\end{proof}

The decoy mechanism is cheat-sensitive rather than fully verifiable. It detects deviations that disturb selected test queries, with probability determined by the decoy rate and test sensitivity. A malicious server may still evade detection on untested ordinary queries with a probability captured by the decoy detection bound. Full verification of delegated QRAM execution would require additional authentication or verification machinery.

%% file: results.tex
\section{Protocol Resource Estimation and Deployment Discussion}
\label{sec:results}

This section summarizes the resource requirements and deployment implications of the protected QRAM protocol. The same accounting model is used throughout: \(N=2^n\), each logical data block has \(m\) bits, the ciphertext bus has \(m+\tau\) qubits, and the server implements a bucket-brigade QRAM over an \(n\)-qubit protected address register and an \((m+\tau)\)-qubit ciphertext bus. Detailed derivations of the primitive costs and table entries are deferred to Appendix~\ref{sec:detailed-resource-estimation}.

\begin{table*}[t]
\centering
\small
\renewcommand{\arraystretch}{1.15}
\setlength{\tabcolsep}{4.0pt}
\begin{tabular}{@{}C{2.25cm}|C{2.8cm}|C{0.8cm}|C{2.8cm}|C{1.2cm}|C{2.15cm}|C{1.4cm}|C{1.3cm}@{}}
\hline\hline
\multirow{2}{*}{\textbf{Scheme}}
& \multicolumn{2}{c|}{\textbf{Qubit count}}
& \multicolumn{2}{c|}{\textbf{Online query depth}}
& \multirow{2}{*}{\makecell[c]{\textbf{Classical}\\\textbf{comm.}}}
& \multirow{2}{*}{\makecell[c]{\textbf{Quantum}\\\textbf{comm.}}}
& \multirow{2}{*}{\makecell[c]{\textbf{Refresh}\\\textbf{freq.}}} \\
\cline{2-5}
& \textbf{Client} & \textbf{Server} & \textbf{Client} & \textbf{Server} & & & \\
\hline\hline

\textbf{qPRP + \(Z\)-padding}
& \makecell[c]{\(n+m+\tau+\)\\\(d_A \log q_A+d_E \log q_E\)}
& \(N\)
& \makecell[c]{\(n d_A(\log q_A)^2+\)\\\((m+\tau)d_E(\log q_E)^2\)}
& \(n+m+\tau\)
& \(\dfrac{N(m+\tau)}{t}\)
& \(n+m+\tau\)
& \makecell[c]{every\\\(t\) queries} \\
\hline

\textbf{qOTP}
& \makecell[c]{\(n+m+\tau+\)\\\(d_E \log q_E\)}
& \(N\)
& \((m+\tau)d_E(\log q_E)^2\)
& \(n+m+\tau\)
& \(N(m+\tau)\)
& \(n+m+\tau\)
& \makecell[c]{every\\query} \\
\hline

\textbf{qPRP + decoys}
& \makecell[c]{\(n+m+\tau+\)\\\(d_A \log q_A+d_E \log q_E\)}
& \(N\)
& \makecell[c]{\(n d_A(\log q_A)^2+\)\\\((m+\tau)d_E(\log q_E)^2\)}
& \(n+m+\tau\)
& \(\dfrac{1}{1-p_{\mathrm{dec}}}\dfrac{N(m+\tau)}{t}\)
& \(\dfrac{n+m+\tau}{1-p_{\mathrm{dec}}}\)
& \makecell[c]{every\\\(t\) queries} \\
\hline

\textbf{qOTP + decoys}
& \makecell[c]{\(n+m+\tau+\)\\\(d_E \log q_E\)}
& \(N\)
& \((m+\tau)d_E(\log q_E)^2\)
& \(n+m+\tau\)
& \(\dfrac{1}{1-p_{\mathrm{dec}}}N(m+\tau)\)
& \(\dfrac{n+m+\tau}{1-p_{\mathrm{dec}}}\)
& \makecell[c]{every\\query} \\
\hline\hline
\end{tabular}
\caption{Summary of resource costs of the protected QRAM variants. Here \(N=2^n\), \(t< O(N^{1/12})\) is the qPRP refresh period, and \(p_{\mathrm{dec}}\) is the decoy probability. All entries report asymptotic scaling up to constant factors and lower-order terms. The two-round protocols introduced in Sec.~\ref{sec:protocol} have the same asymptotic scaling as the one-round protocols. Communication entries are amortized per online query unless otherwise stated.}
\label{tab:scheme-summary}
\end{table*}

\subsection{Resource Model and Protocol Costs}
\label{subsec:resource-cost}

The address-hiding layer is instantiated by a seven-round Feistel qPRP over \(\{0,1\}^n\), while the encryption layer uses an independent seven-round Feistel qPRP over the \((m+\tau)\)-bit ciphertext space. The address-side implementation is parameterized by \((d_A,q_A)\), and the encryption-side implementation is parameterized by \((d_E,q_E)\). Under the reversible arithmetic convention in Appendix~\ref{subsec:detailed-resource-model}, coherent address permutation contributes the \(d_A,q_A\) terms in Table~\ref{tab:scheme-summary}, while coherent encryption and decryption contribute the \(d_E,q_E\) terms.

On the server side, the protected protocol preserves the standard bucket-brigade QRAM scaling: the server stores the protected layout using \(\Theta(N)\) QRAM resources, and one protected lookup has depth \(\Theta(n+m+\tau)\). The online quantum communication is therefore the protected address plus ciphertext bus. The qPRP-based protocol additionally requires coherent address permutation on the client, whereas the qOTP-based variant replaces this step with Pauli masking and removes the address-permutation term from the online client depth.

Table~\ref{tab:scheme-summary} summarizes the resulting online and amortized costs. The dominant distinction is not the server-side QRAM cost, which remains the same across variants, but the refresh amortization. A qPRP-protected layout can serve \(t<O(N^{1/12})\) protected query rounds before rerandomization, while a qOTP-shifted layout is single-use. Decoy variants preserve the same cost per executed round and only rescale amortized per-real-query communication by \(1/(1-p_{\mathrm{dec}})\). The full cost derivation is given in Appendix~\ref{subsec:detailed-protocol-cost}.

\subsection{Refresh and Synchronization Options}
\label{subsec:refresh-sync}

A refresh synchronizes the server-side protected layout with the client's current keys, randomness, and logical database contents. Refresh may be required because the application updates the database, or because the security epoch expires even when the database itself is unchanged. This is the main deployment tradeoff between the qPRP and qOTP variants: qPRP is preferable when many queries reuse the same database layout, since the rebuild cost is amortized over an epoch; qOTP is more attractive when the application already forces a fresh layout for each query, since its lighter online masking avoids coherent address qPRP evaluation.

Among all the database refreshing methods, the baseline is a full client rebuild. At the start of a new epoch, the client samples fresh protection material, re-encrypts all \(N\) blocks with fresh randomness, places each ciphertext at its new hidden physical location, and uploads the protected layout. This requires \(N(m+\tau)\) bits of classical upload per refresh, but introduces no additional helper, trust assumption, or cryptographic primitive beyond those already used by the protocol.

Other synchronization mechanisms can be substituted when deployment constraints make full client rebuild undesirable. MPC, FHE, trusted execution, or server-assisted oblivious shuffling may reduce client-side work or improve system flexibility, at the cost of additional trust or cryptographic assumptions~\cite{shriram2023ruffle,song2024secretshuffle,gentry2009fhe,chillotti2020tfhe,brakerski2014leveled,ohrimenko2014melbourne,patel2018cacheshuffle}. ORAM-style mechanisms are best viewed as sparse update layers rather than global refresh mechanisms~\cite{stefanov2013path,ren2015ring}: if only a small number of entries change while the current qPRP epoch remains valid, such a layer can hide which logical addresses were updated while moving less data than a full rebuild. If most entries must be rerandomized, or if a new security epoch requires fresh ciphertext randomness and a fresh hidden layout, full rebuild is usually more efficient. In this work, any such mechanism must instantiate its encryption, permutation, commitment, or authentication primitives with quantum-secure versions consistent with the OQRAM model. A more detailed comparison with helper-assisted and ORAM-style synchronization mechanisms is deferred to Appendix~\ref{subsec:layout-update-options} and summarized in Table~\ref{tab:layout-update-protocols}.

\subsection{Comparison with Blind Quantum Computing}
\label{subsec:comparison-ubqc}

Though one may use a fully blind quantum-computing approach such as UBQC to protect an entire delegated computation, applying it to QRAM would require blinding the full QRAM circuit or measurement pattern, resulting in much larger overhead for both computation and communication. For a QRAM lookup circuit with depth \(\Theta(\log(n+m))\) and size \(\Theta(N)\), a UBQC-style implementation requires the client to prepare and transmit \(\Theta(N)\) random single-qubit states and to exchange \(\Theta(N)\) classical measurement messages per query, with \(\Theta(\log(n+m))\) adaptive measurement rounds~\cite{broadbent2009universal}. OQRAM instead protects only the delegated quantum-query interface. Under the bucket-brigade model used here, the qPRP-based scheme keeps the server-side QRAM footprint at \(\Theta(N)\), uses only \(\Theta(n+m+\tau)\) online quantum communication per query, and amortizes refresh to \(\Theta(N(m+\tau)/t)\) classical bits per real query. Thus, OQRAM does not provide full UBQC-style blindness, but avoids the \(\Theta(N)\) per-query blind-computation communication cost when the target primitive is private delegated QRAM access.

%% file: conclusion.tex
\section{Conclusion}
\label{sec:conclusion}
This work presents oblivious QRAM, a protocol framework for securing delegated coherent quantum queries between a lightweight \(O(n+m+anc)\)-qubit client and an \(O(N)\)-scale QRAM server. The protocol combines coherent qPRP-based address relabeling with an encrypted and reshuffled database layout, hiding both the logical address state and plaintext database contents while preserving server-side QRAM lookup. Moreover, a qOTP-based variant provides information-theoretic single-query address hiding with per-query refresh, and hidden decoy queries add cheat-sensitive detection of malicious behavior. Oblivious QRAM therefore gives a query-specific alternative to fully blind quantum computation, while largely reducing quantum communication and client-side quantum resource requirements.

%% file: appendix.tex
\section{Detailed Protocol Descriptions}
\label{sec:detailed-protocol}

\subsection{Two-round query-use-unquery state evolution}
\label{subsec:two-round-detailed}

This appendix expands the two-round wrapper from Sec.~\ref{subsec:two-round}. The client starts with address register $A$, a clean ciphertext bus $B$, and a local algorithmic register $M$:
\[
\ket{\psi}_A \ket{0^{m+\tau}}_B \ket{b_i}_{M}
= \sum_i \alpha_i \ket{i}_A \ket{0^{m+\tau}}_B \ket{b_i}_{M}.
\]

\begin{enumerate}
    \item
    The client applies the same address protection map as in the one-round protocol to \(A\), producing
    \[
    \ket{\widetilde{\psi}^\star}_A
    =
    \sum_i \alpha_i (-1)^{z\cdot i}\ket{\pi^\star(i)}_A .
    \]
    The client sends \(A\) and \(B\) to the server, which applies the protected oracle \(O_{C^\star}\) and returns the registers. The resulting state is
    \begin{align*}
    &\sum_i \alpha_i (-1)^{z\cdot i}
    \ket{\pi^\star(i)}_A
    \ket{C^\star[\pi^\star(i)]}_B
    \ket{b_i}_{M} \\
    &\quad =
    \sum_i \alpha_i (-1)^{z\cdot i}
    \ket{\pi^\star(i)}_A
    \ket{C[i]}_B
    \ket{b_i}_{M},
    \end{align*}
    where the equality follows from the layout invariant
    \[
    C^\star[\pi^\star(i)]=C[i].
    \]

    \item
    The client applies the inverse address-protection map to recover the logical address labels,
    \[
    \sum_i \alpha_i \ket{i}_A \ket{C[i]}_B \ket{b_i}_{M}.
    \]
    The client then applies the inverse encryption permutation \((\pi'_{K_{\mathrm{enc}}})^{-1}\) coherently to the bus register \(B\), obtaining
    \[
    \sum_i \alpha_i \ket{i}_A \ket{D[i]\|r_i}_B \ket{b_i}_{M}.
    \]
    The \(\tau\)-qubit randomness suffix is kept, since the bus will later be re-encrypted and unqueried.

    \item
    The client applies bitwise CNOTs from the first \(m\) qubits of \(B\), which contain $D[i]$, into the local work register \(M\). This gives
    \[
    \sum_i \alpha_i \ket{i}_A \ket{D[i]\|r_i}_B \ket{b_i\oplus D[i]}_{M}.
    \]
    At this point, the desired work-register update has been performed, while the full expanded block remains available in \(B\) for uncomputation.

    \item
    The client reapplies \(\pi'_{K_{\mathrm{enc}}}\) to the bus register \(B\), recovering
    \[
    \sum_i \alpha_i \ket{i}_A \ket{C[i]}_B \ket{b_i\oplus D[i]}_{M}.
    \]
    The client then reapplies the same address protection map as before, so that the address register is again expressed in the protected basis expected by the server:
    \[
    \sum_i \alpha_i (-1)^{z\cdot i}
    \ket{\pi^\star(i)}_A
    \ket{C^\star[\pi^\star(i)]}_B
    \ket{b_i\oplus D[i]}_{M}.
    \]

    \item
    The client keeps \(M\) locally and sends \((A,B)\) back to the server. The server applies the same protected oracle once more. Because the bus already contains the same ciphertext block loaded in the first query, this second call erases the bus:
    \[
    \begin{aligned}
    &\sum_i \alpha_i (-1)^{z\cdot i}
    \ket{\pi^\star(i)}_A
    \ket{C^\star[\pi^\star(i)]}_B
    \ket{b_i\oplus D[i]}_{M} \\
    &\xmapsto{\,O_{C^\star}\,}
    \sum_i \alpha_i (-1)^{z\cdot i}
    \ket{\pi^\star(i)}_A
    \ket{0^{m+\tau}}_B
    \ket{b_i\oplus D[i]}_{M}.
    \end{aligned}
    \]
    The server then returns \((A,B)\) to the client.

    \item
    Finally, the client removes the address protection and obtains
    \[
    \sum_i \alpha_i \ket{i}_A \ket{0^{m+\tau}}_B \ket{b_i\oplus D[i]}_{M}.
    \]
\end{enumerate}

\subsection{qOTP-based single-query masking steps}
\label{subsec:qotp-detailed}

This appendix expands the qOTP-based variant from Sec.~\ref{subsec:qotp-based}. For each independent query, the client prepares a fresh shifted encrypted layout as follows.
\begin{enumerate}
    \item The client samples fresh qOTP address-mask strings
    \[
    \mathbf{x},\mathbf{z}\xleftarrow{\$}\{0,1\}^n,
    \]
    where \(\mathbf{x}\) is the computational-basis shift mask and \(\mathbf{z}\) is the phase mask.
    \item The client samples a fresh encryption key set \(K_{\mathrm{enc}}\).
    \item For each logical address $i\in\{0,1\}^n$, the client samples independent fresh block randomness \(r_i\xleftarrow{\$}\{0,1\}^\tau\).
    \item The client computes the ciphertext for each address:
    \[
    C[i]=\mathrm{Enc}_{K_{\mathrm{enc}}}(D[i];r_i)
    =
    \pi'_{K_{\mathrm{enc}}}(D[i]\|r_i).
    \]
    \item The client constructs the physical layout by shifting the encrypted logical database according to the qOTP \(X\)-mask:
    \[
    C^{\mathbf{x}}[j]=C[j\oplus \mathbf{x}].
    \]
    Equivalently, \(C^{\mathbf{x}}[i\oplus \mathbf{x}]=C[i]\) for every logical address \(i\).
    \item The client uploads \(C^\mathbf{x}\) to the server, and the server uses this fresh database layout for only the current independent QRAM query.
\end{enumerate}

The corresponding online query proceeds as follows.
\begin{enumerate}
    \item For an input address state \(\ket{\psi}_A=\sum_i \alpha_i \ket{i}_A\), the client applies the fresh qOTP mask:
    \[
        \ket{\widetilde{\psi}}_A
        =
        X^{\mathbf{x}}Z^{\mathbf{z}}\ket{\psi}_A
        =
        \sum_i \alpha_i (-1)^{\mathbf{z}\cdot i}
        \ket{i\oplus \mathbf{x}}_A.
    \]
    The client sends \(\ket{\widetilde{\psi}}_A\) to the server.
    \item The server evaluates QRAM over \(C^{\mathbf{x}}\), obtaining
    \[
    \begin{aligned}
        &\sum_i \alpha_i (-1)^{\mathbf{z}\cdot i}
        \ket{i\oplus\mathbf{x}}_A
        \ket{C^{\mathbf{x}}[i\oplus\mathbf{x}]}_B \\
        &\qquad =
        \sum_i \alpha_i (-1)^{\mathbf{z}\cdot i}
        \ket{i\oplus\mathbf{x}}_A
        \ket{C[i]}_B.
    \end{aligned}
    \]
    The server sends the address-ciphertext registers back to the client.
    \item The client applies the inverse qOTP mask \(Z^{\mathbf{z}}X^{\mathbf{x}}\), up to a global phase, and recovers the logical address labels:
    \[
    \begin{aligned}
    &Z^{\mathbf{z}}X^{\mathbf{x}}
    \left(
        \sum_i \alpha_i (-1)^{\mathbf{z}\cdot i}
        \ket{i\oplus\mathbf{x}}_A\ket{C[i]}_B
    \right) \\
    &\qquad =
    \sum_i \alpha_i (-1)^{\mathbf{z}\cdot i}
    (-1)^{\mathbf{z}\cdot i}
    \ket{i}_A\ket{C[i]}_B
    =
    \sum_i \alpha_i \ket{i}_A \ket{C[i]}_B.
    \end{aligned}
    \]
    The ciphertext recovery and plaintext-use steps then proceed exactly as in the qPRP-based protocol:
    \[
        (\pi'_{K_{\mathrm{enc}}})^{-1}\ket{C[i]}_B
        =
        \ket{D[i]\|r_i}_B.
    \]
\end{enumerate}

\section{Detailed Proof}\label{sec:detailed-proof}

\subsection{Honest-but-curious qPRP-based single query privacy}
\label{subsec:hbc-qprp-proof}
The proof proceeds by a sequence of hybrids.

\paragraph{Hybrid \(H_0\): Real execution.}
This is the real single-query protocol. The server stores the encrypted
permuted layout
\[
    C^\pi[j]
    =
    C[\pi_{K_{\mathrm{prp}}}^{-1}(j)]
\]
and receives the protected address register. For a logical query state
\[
    |\psi\rangle_A
    =
    \sum_{i\in\{0,1\}^n}\alpha_i |i\rangle_A,
\]
the address register sent to the server is, before QRAM lookup,
\[
    U_{\pi} Z^z |\psi\rangle_A
    =
    \sum_i \alpha_i (-1)^{z\cdot i}
    |\pi_{K_{\mathrm{prp}}}(i)\rangle_A,
\]
where \(z\leftarrow\{0,1\}^n\) is fresh if \(Z\)-padding is enabled. The server
also has access to the classical encrypted layout \(C^\pi\), but it does not
know \(K_{\mathrm{prp}}\), \(K_{\mathrm{enc}}\), or the encryption randomness
\(\{r_i\}\).

\paragraph{Hybrid \(H_1\): Replace the address qPRP by a random permutation.}
In \(H_1\), replace the keyed permutation
\[
    \pi_{K_{\mathrm{prp}}}
\]
by a uniformly random permutation
\[
    P:\{0,1\}^n\rightarrow\{0,1\}^n.
\]
The encrypted layout is permuted consistently with \(P\), and the online
address register is prepared using \(U_P\). By the computational address-hiding
lemma, any efficient server that distinguishes \(H_0\) from \(H_1\) with
advantage greater than
\[
    \varepsilon_{\mathsf{addr}}(q_{\mathsf{addr}},\lambda)
\]
would give a distinguisher against the strong qPRP security of
\(\pi_{K_{\mathrm{prp}}}\). Therefore,
\[
    H_0 \approx_{\varepsilon_{\mathsf{addr}}} H_1.
\]

\paragraph{Hybrid \(H_2\): Dephase the address register.}
In \(H_1\), the protected address state is
\[
    \sum_i \alpha_i (-1)^{z\cdot i}|P(i)\rangle.
\]
Averaging over fresh uniform \(z\), the server's address density matrix becomes
\[
    \mathbb{E}_{z}
    \left[
        U_P Z^z |\psi\rangle\langle\psi| Z^z U_P^\dagger
    \right].
\]
By the \(Z\)-padding dephasing lemma,
\[
    \mathbb{E}_{z}
    \left[
        Z^z |\psi\rangle\langle\psi| Z^z
    \right]
    =
    \sum_i |\alpha_i|^2 |i\rangle\langle i|.
\]
Therefore,
\[
    \mathbb{E}_{z}
    \left[
        U_P Z^z |\psi\rangle\langle\psi| Z^z U_P^\dagger
    \right]
    =
    \sum_i |\alpha_i|^2 |P(i)\rangle\langle P(i)|.
\]
Thus \(Z\)-padding removes the phase information in the query state. The
remaining address-side information is only the diagonal distribution
\[
    \{|\alpha_i|^2:i\in\{0,1\}^n\},
\]
but attached to physical labels through the hidden random permutation \(P\).

\paragraph{Hybrid \(H_3\): Hide the logical labels by the random permutation.}
Because \(P\) is uniformly random and independent of the query state, the
physical label \(P(i)\) is a uniformly random relabeling of the logical label
\(i\). Hence the server cannot associate a particular weight
\(|\alpha_i|^2\) with the corresponding logical address \(i\). The address
state in this hybrid has the form
\[
    \sum_i |\alpha_i|^2 |P(i)\rangle\langle P(i)|.
\]
Equivalently, it is a diagonal state whose eigenvalue multiset is
\[
    \{|\alpha_i|^2:i\in\{0,1\}^n\},
\]
but whose basis labels have been randomly permuted. Therefore, the simulator
only needs the relabeling-invariant dephased query distribution
\[
    \{|\alpha_i|^2:i\in\{0,1\}^n\},
\]
rather than the logical-to-physical correspondence. This is precisely the
address-related part of the leakage
\[
    \mathcal{L}
    =
    \left(
        N,\,
        m+\tau,\,
        1,\,
        \{|\alpha_i|^2:i\in\{0,1\}^n\}
    \right).
\]

\paragraph{Hybrid \(H_4\): Replace the encrypted database by dummy encryptions.}
Now replace the encrypted layout of the real database \(D\) by an encrypted
layout of the fixed dummy database
\[
    D^\star[i]=0^m
    \qquad
    \text{for all } i.
\]
That is, replace each real ciphertext
\[
    C[i]
    =
    \mathrm{Enc}_{K_{\mathrm{enc}}}(D[i];r_i)
\]
by
\[
    C^\star[i]
    =
    \mathrm{Enc}_{K_{\mathrm{enc}}}(0^m;r_i^\star),
\]
with fresh randomness \(r_i^\star\), and place the dummy ciphertexts according
to the same hidden physical permutation.

By database confidentiality, the real encrypted layout and the dummy encrypted
layout are computationally indistinguishable. More explicitly, define hybrids
\[
    G_0,G_1,\ldots,G_N,
\]
where \(G_\ell\) encrypts dummy blocks for the first \(\ell\) logical positions
and real database blocks for the remaining \(N-\ell\) positions. Adjacent
hybrids \(G_\ell\) and \(G_{\ell+1}\) differ in only one ciphertext. If an
efficient adversary distinguished adjacent hybrids with non-negligible
advantage, then one could build a qIND-qCPA adversary that embeds its challenge
ciphertext at the differing position and simulates all other ciphertexts
honestly. Hence each adjacent transition costs at most
\[
    \varepsilon_{\mathsf{enc}}(\lambda),
\]
and the full replacement costs at most
\[
    N\cdot \varepsilon_{\mathsf{enc}}(\lambda).
\]
Therefore,
\[
    H_3 \approx_{N\varepsilon_{\mathsf{enc}}} H_4.
\]

\paragraph{Hybrid \(H_5\): Ideal simulation.}
In \(H_4\), the server's view consists of:
\begin{enumerate}
    \item the public number of cells \(N\);
    \item the ciphertext block length \(m+\tau\);
    \item the fact that exactly one lookup occurred;
    \item the relabeling-invariant dephased query distribution
    \(\{|\alpha_i|^2:i\in\{0,1\}^n\}\);
    \item an encrypted dummy layout independent of the real database; and
    \item a randomly relabeled address register independent of the logical
    address labels.
\end{enumerate}
This view can be generated by a simulator given only
\[
    \mathcal{L}
    =
    \left(
        N,\,
        m+\tau,\,
        1,\,
        \{|\alpha_i|^2:i\in\{0,1\}^n\}
    \right).
\]
The simulator samples a random permutation \(P\), constructs an encrypted dummy
layout, and prepares a diagonal address state with eigenvalue multiset
\(\{|\alpha_i|^2\}\) assigned to physical labels according to \(P\). Because
the real database contents and the logical-to-physical address correspondence
have both been removed in the preceding hybrids, the simulated view is
identical to \(H_4\) up to the already accounted-for computational losses.

Combining the hybrid steps gives total distinguishing advantage at most
\[
    \varepsilon_{\mathsf{addr}}(q_{\mathsf{addr}},\lambda)
    +
    N\cdot \varepsilon_{\mathsf{enc}}(\lambda),
\]
as claimed.

\subsection{Honest-but-curious qOTP-based single query privacy}
\label{subsec:hbc-qotp-proof}
The proof proceeds by hybrids.

\paragraph{Hybrid \(H_0\): Real execution.}
This is the real qOTP-based protocol. The server receives the shifted encrypted
layout \(C^a\) and the qOTP-protected address register. For a query state
\[
    |\psi\rangle_A
    =
    \sum_i \alpha_i |i\rangle_A,
\]
the address register sent to the server is
\[
    X^aZ^b|\psi\rangle_A
    =
    \sum_i \alpha_i(-1)^{b\cdot i}|i\oplus a\rangle_A.
\]

\paragraph{Hybrid \(H_1\): Replace the protected address by maximally mixed.}
By Lemma~\ref{lem:qotp-address-hiding}, averaging over fresh uniform
\(a,b\) gives
\[
    \mathbb E_{a,b}
    \left[
        X^aZ^b|\psi\rangle\langle\psi|Z^bX^a
    \right]
    =
    \frac{I_A}{N}.
\]
Thus the address-side view is information-theoretically independent of the
logical query state, including its amplitudes and phases. Unlike the qPRP-based
variant, the ideal leakage does not need to include
\(\{|\alpha_i|^2\}\), because the fresh \(X\)-pad uniformly shifts the diagonal
distribution.

\paragraph{Hybrid \(H_2\): Replace the encrypted layout by dummy encryptions.}
Replace the shifted encrypted layout of the real database by a shifted encrypted
layout of a dummy database \(D^\star[i]=0^m\). As in
Lemma~\ref{lem:database-confidentiality}, use \(N\) hybrids, replacing one
encrypted block at a time. Each adjacent hybrid differs in one ciphertext and
is indistinguishable by qIND-qCPA security. Therefore the total loss is at most
\[
    N\cdot \varepsilon_{\mathsf{enc}}(\lambda).
\]

\paragraph{Hybrid \(H_3\): Ideal simulation.}
After the previous hybrids, the server's view consists only of a maximally
mixed address register, an encrypted dummy layout, and the public parameters
\(N\), \(m+\tau\), and the fact that one lookup occurred. A simulator given
\[
    \mathcal{L}_{\mathsf{qOTP}}
    =
    (N,m+\tau,1)
\]
can generate this view directly. Hence the real and ideal views are
indistinguishable with advantage at most
\(N\cdot\varepsilon_{\mathsf{enc}}(\lambda)\).

\begin{table*}[t]
\centering
\small
\renewcommand{\arraystretch}{1.15}
\setlength{\tabcolsep}{4.0pt}
\begin{tabular}{@{}C{2.5cm}C{2.5cm}C{3.5cm}C{4cm}C{3.5cm}@{}}
\hline\hline
\textbf{Mechanism}
& \textbf{Communication}
& \textbf{Client / helper memory}
& \textbf{Computation}
& \textbf{Best use case} \\
\hline\hline

\textbf{\makecell[c]{Client rebuild\\(baseline)}}
& \(\Theta(BN)\)
& Streaming to linear client memory
& Client: \(\Theta(N)\) encryptions and permutation evaluations; server: linear rewrite
& Clean theorem baseline \\
\hline

\textbf{MPC helper~\cite{shriram2023ruffle,song2024secretshuffle}}
& \(\Theta(BN)\) w/ preprocessing
& Distributed across helper parties
& Online cost can be linear; offline preprocessing remains substantial
& Weak client with strong external infrastructure \\
\hline

\textbf{HE / FHE helper~\cite{gentry2009fhe,chillotti2020tfhe,brakerski2014leveled}}
& Generic encrypted reshuffle \(\Theta(BN^2)\)
& Small client memory; heavy helper/server cryptographic state
& High helper/server homomorphic computation
& Specialized deployments with very weak clients \\
\hline

\textbf{\makecell[c]{Melbourne-style\\shuffle~\cite{ohrimenko2014melbourne}}}
& \(O(BN)\) per full reshuffle
& \(O(\sqrt{N}B)\) private client memory
& Linear orchestration and data movement
& ORAM-style reshuffle with moderate client scratch space \\
\hline

\textbf{Cache Shuffle-style shuffle~\cite{patel2018cacheshuffle}}
& \((4+\epsilon)NB\) for CacheShuffleRoot
& \(O(\sqrt{N}B)\) for CacheShuffleRoot
& Client-side oblivious shuffle with tunable bandwidth-memory tradeoff
& Bandwidth-sensitive ORAM-style deployment \\
\hline

\textbf{\makecell[c]{Classical\\ORAM layer~\cite{goldreich1996software,stefanov2013path,ren2015ring}}}
& \(O(B\cdot \mathrm{polylog}(N))\) per update
& Polylogarithmic to small client stash/position state
& Per-update reshuffling and re-encryption
& Sparse application-driven updates within a qPRP epoch \\
\hline\hline

\end{tabular}
\caption{Candidate mechanisms for refreshing or updating the server-side protected layout. \(B=m+\tau\) denotes the ciphertext block size.}
\label{tab:layout-update-protocols}
\end{table*}

\section{Detailed resource estimation and comparison}
\label{sec:detailed-resource-estimation}

\subsection{Primitive instantiation and accounting model}
\label{subsec:detailed-resource-model}
This appendix expands the resource model used in Sec.~\ref{subsec:resource-cost}. Throughout the protocol description in Sec.~\ref{sec:protocol}, the address permutation and encryption were treated as keyed coherent operations \(U^\pi\). For resource accounting, the address-hiding permutation is instantiated by a balanced seven-round Feistel qPRP over \(\{0,1\}^n\), with round functions implemented by a BPR-style ring-based qPRF~\cite{banerjee2012prf,banerjee2014spring,carolan2025compressed}. The encryption permutation \(\pi'_{K_{\mathrm{enc}}}\) is implemented independently over the \((m+\tau)\)-bit ciphertext space~\cite{lyubashevsky2013ringlwe,banerjee2012prf,banerjee2014spring}.

For the address-side implementation, let \(d_A\) and \(q_A\) denote the ring dimension and modulus, and set
\[
    b_A:=\lceil \log_2 q_A\rceil .
\]
Let \(D_{\mathrm{mul}}^{(A)}(b_A)\) be the depth of one reversible modular multiplication on \(b_A\)-bit operands. Under the conservative arithmetic model used here,
\[
    D_{\mathrm{mul}}^{(A)}(b_A)=\Theta(b_A^2)=\Theta((\log q_A)^2).
\]
The additional workspace and client-side depth for one coherent evaluation of \(U_\pi\) are denoted by \(Q_A\) and \(D_A\). For the above qPRP instantiation,
\[
    Q_A=\Theta\!\left(\frac{n}{2}+d_A\log q_A\right),
\]
and
\[
    D_A
    =\Theta\!\left(n\,d_A\,D_{\mathrm{mul}}^{(A)}(b_A)\right)
    =\Theta\!\left(n\,d_A\,(\log q_A)^2\right).
\]

The qPRP on the client's address qubits is implemented coherently under a standard parallel-depth convention for reversible arithmetic circuits. In particular, the depth estimates allow multi-qubit \texttt{TOFFOLI}, multi-qubit \texttt{FANOUT}, and mid-circuit measurement with feedforward. This is an accounting convention for circuit depth and is separate from the cryptographic assumptions used for qPRP and encryption security.

For the encryption-side implementation, let \(d_E\) and \(q_E\) denote the ring dimension and modulus, set
\[
    b_E:=\lceil \log_2 q_E\rceil ,
\]
and define \(D_{\mathrm{mul}}^{(E)}(b_E)\) analogously. Under the same multiplication model,
\[
    D_{\mathrm{mul}}^{(E)}(b_E)=\Theta((\log q_E)^2).
\]
The additional workspace and client-side depth for one coherent evaluation of \(U_{\pi'_{K_{\mathrm{enc}}}}\) or \(U_{\pi'_{K_{\mathrm{enc}}}}^{-1}\) are denoted by \(Q_E\) and \(D_E\). The accounting uses
\[
    Q_E=\Theta\!\left(\frac{m+\tau}{2}+d_E\log q_E\right),
\]
and
\[
    D_E
    =\Theta\!\left((m+\tau)\,d_E\,D_{\mathrm{mul}}^{(E)}(b_E)\right)
    =\Theta\!\left((m+\tau)\,d_E\,(\log q_E)^2\right).
\]

For the server-side QRAM, the bucket-brigade QRAM model is used as the baseline implementation~\cite{giovannetti2008quantum}. For an \(n\)-qubit address register and an \(m\)-qubit data bus, a bucket-brigade query uses \(\Theta(N)\) server-side qubit overhead and has query depth \(\Theta(n+m)\). In the protected protocol, the server operates on an \(n\)-qubit protected address register and an \((m+\tau)\)-qubit ciphertext bus. Therefore, for Table~\ref{tab:scheme-summary}, the costs are
\[
    Q_{\mathrm{QRAM}}=\Theta(N),
    \qquad
    D_{\mathrm{QRAM}}=\Theta(n+m+\tau).
\]

\subsection{Derivation of the protocol-cost table}
\label{subsec:detailed-protocol-cost}
This subsection expands the accounting behind Table~\ref{tab:scheme-summary}. The table substitutes the primitive-level quantities \(D_A,Q_A,D_E,Q_E\) with their asymptotic scaling and suppresses constant factors and lower-order terms.

For the qPRP-based protocol, the client coherently evaluates the address permutation before sending the query, applies the inverse address permutation after receiving the server response, and then coherently inverts the encryption permutation on the ciphertext bus. Up to constant factors from the forward and inverse address maps, the online client depth is
\[
    \Theta(D_A+D_E),
\]
and the client workspace is
\[
    \Theta(Q_A+Q_E+n+m+\tau).
\]
Substituting the primitive costs from Appendix~\ref{subsec:detailed-resource-model} gives the qPRP row of Table~\ref{tab:scheme-summary}.

For the qOTP-based protocol, Pauli \(X\)- and \(Z\)-masks replace coherent address-permutation evaluation. These masks contribute only \(O(n)\) single-qubit gates, so the online client depth is dominated by ciphertext recovery:
\[
    \Theta(D_E).
\]
The client workspace is
\[
    \Theta(Q_E+n+m+\tau).
\]
Again substituting the encryption primitive cost gives the qOTP row of Table~\ref{tab:scheme-summary}.

In both protocols, a protected online query sends the \(n\)-qubit address register and the \((m+\tau)\)-qubit ciphertext bus between the client and server, so the quantum communication per protected query is
\[
    \Theta(n+m+\tau).
\]
The server's online cost is the protected QRAM lookup, with qubit overhead \(\Theta(N)\) and depth \(\Theta(n+m+\tau)\).

The classical communication differs because the two layouts have different reuse properties. In the qPRP-based protocol, one encrypted and shuffled layout can be reused for \(t\) protected query rounds, so a full \(N\)-block layout upload contributes amortized classical communication
\[
    \Theta\!\left(\frac{N(m+\tau)}{t}\right)
\]
per real query. In the qOTP-based protocol, the shifted layout must be refreshed after every independent logical query, giving
\[
    \Theta(N(m+\tau))
\]
classical communication per real query.

Finally, decoy queries only affect amortization. Under independent decoy sampling with probability \(p_{\mathrm{dec}}\), only a \(1-p_{\mathrm{dec}}\) fraction of executed rounds are real queries in expectation. Therefore, costs proportional to executed rounds are multiplied by \(1/(1-p_{\mathrm{dec}})\) when reported per real query. This factor applies to online latency, server-side QRAM latency, quantum communication, and refresh communication.

\subsection{Refresh and synchronization mechanism comparison}
\label{subsec:layout-update-options}
Table~\ref{tab:layout-update-protocols} introduces different refresh and synchronization mechanisms and their best use cases.

One can also use background refresh to further reduce practical latency by preparing the next protected layout while the current epoch is still serving queries. This does not change the total communication cost: each online query must still be encoded and decoded using the matching layout, address key, encryption key, and query mask.

%% file: refs.bib
@article{broadbent2015delegating,
  author        = {Broadbent, Anne},
  title         = {Delegating Private Quantum Computations},
  journal       = {Canadian Journal of Physics},
  volume        = {93},
  number        = {9},
  pages         = {941--946},
  year          = {2015},
  doi           = {10.1139/cjp-2015-0030}
}

@article{giovannetti2008architectures,
  title={Architectures for a quantum random access memory},
  author={Giovannetti, Vittorio and Lloyd, Seth and Maccone, Lorenzo},
  journal={Physical Review A},
  volume={78},
  number={5},
  pages={052310},
  year={2008},
  publisher={APS}
}

@article{giovannetti2008quantum,
  title={Quantum random access memory},
  author={Giovannetti, Vittorio and Lloyd, Seth and Maccone, Lorenzo},
  journal={Physical review letters},
  volume={100},
  number={16},
  pages={160501},
  year={2008},
  publisher={APS}
}

@article{giovannetti2008privatequery,
  title={Quantum private queries},
  author={Giovannetti, Vittorio and Lloyd, Seth and Maccone, Lorenzo},
  journal={Physical Review Letters},
  volume={100},
  number={23},
  pages={230502},
  year={2008},
  publisher={APS}
}

@inproceedings{broadbent2009universal,
  title={Universal blind quantum computation},
  author={Broadbent, Anne and Fitzsimons, Joseph and Kashefi, Elham},
  booktitle={2009 50th Annual IEEE Symposium on Foundations of Computer Science},
  pages={517--526},
  year={2009},
  organization={IEEE}
}

@article{goldreich1996software,
  title={Software protection and simulation on oblivious RAMs},
  author={Goldreich, Oded and Ostrovsky, Rafail},
  journal={Journal of the ACM},
  volume={43},
  number={3},
  pages={431--473},
  year={1996},
  publisher={ACM}
}

@inproceedings{asharov2020optorama,
  title={OptORAMa: optimal oblivious RAM},
  author={Asharov, Gilad and Komargodski, Ilan and Lin, Wei-Kai and Nayak, Kartik and Shi, Elaine},
  booktitle={Advances in Cryptology -- EUROCRYPT 2020},
  pages={221--251},
  year={2020},
  organization={Springer}
}

@misc{tople2018prooram,
  title={PRO-ORAM: Constant latency read-only oblivious RAM},
  author={Tople, Shruti and Jia, Yaoqi and Saxena, Prateek},
  year={2018},
  note={Cryptology ePrint Archive}
}

@inproceedings{grover1996fast,
  title={A fast quantum mechanical algorithm for database search},
  author={Grover, Lov K},
  booktitle={Proceedings of the twenty-eighth annual ACM symposium on Theory of computing},
  pages={212--219},
  year={1996}
}

@book{nielsen2010quantum,
  title={Quantum computation and quantum information},
  author={Nielsen, Michael A and Chuang, Isaac L},
  year={2010},
  publisher={Cambridge university press}
}

@article{weiss2024quantum,
  title={Quantum Random Access Memory Architectures Using 3D Superconducting Cavities},
  author={Weiss, DK and Puri, Shruti and Girvin, SM},
  journal={PRX Quantum},
  volume={5},
  number={2},
  pages={020312},
  year={2024},
  publisher={APS}
}

@inproceedings{xu2023systems,
  title={Systems architecture for quantum random access memory},
  author={Xu, Shifan and Hann, Connor T and Foxman, Ben and Girvin, Steven M and Ding, Yongshan},
  booktitle={Proceedings of the 56th Annual IEEE/ACM International Symposium on Microarchitecture},
  pages={526--538},
  year={2023}
}

@article{carolan2025compressed,
  title={Compressed permutation oracles},
  author={Carolan, Joseph},
  journal={arXiv preprint arXiv:2509.18586},
  year={2025}
}

@article{zhandry2025note,
  title={A note on quantum-secure PRPs},
  author={Zhandry, Mark},
  journal={Quantum},
  volume={9},
  pages={1696},
  year={2025},
  publisher={Verein zur F{\"o}rderung des Open Access Publizierens in den Quantenwissenschaften}
}

@inproceedings{gagliardoni2016semantic,
  title={Semantic security and indistinguishability in the quantum world},
  author={Gagliardoni, Tommaso and H{\"u}lsing, Andreas and Schaffner, Christian},
  booktitle={Annual international cryptology conference},
  pages={60--89},
  year={2016},
  organization={Springer}
}

@inproceedings{stefanov2013path,
  author    = {Emil Stefanov and Marten van Dijk and Elaine Shi and Christopher W. Fletcher and Ling Ren and Xiangyao Yu and Srinivas Devadas},
  title     = {Path {ORAM}: An Extremely Simple Oblivious {RAM} Protocol},
  booktitle = {Proceedings of the 2013 ACM SIGSAC Conference on Computer and Communications Security},
  pages     = {299--310},
  year      = {2013},
  doi       = {10.1145/2508859.2516660}
}

@inproceedings{islam2012access,
  author    = {Mohammad Saiful Islam and Mehmet Kuzu and Murat Kantarcioglu},
  title     = {Access Pattern Disclosure on Searchable Encryption: Ramification, Attack and Mitigation},
  booktitle = {Proceedings of the Network and Distributed System Security Symposium},
  year      = {2012}
}

@inproceedings{kellaris2016generic,
  author    = {Georgios Kellaris and George Kollios and Kobbi Nissim and Adam O'Neill},
  title     = {Generic Attacks on Secure Outsourced Databases},
  booktitle = {Proceedings of the 2016 ACM SIGSAC Conference on Computer and Communications Security},
  pages     = {1329--1340},
  year      = {2016},
  doi       = {10.1145/2976749.2978386}
}

@article{childs2005secure,
  author  = {Andrew M. Childs},
  title   = {Secure Assisted Quantum Computation},
  journal = {Quantum Information and Computation},
  volume  = {5},
  number  = {6},
  pages   = {456--466},
  year    = {2005}
}

@article{fitzsimons2017private,
  author  = {Joseph F. Fitzsimons},
  title   = {Private Quantum Computation: An Introduction to Blind Quantum Computing and Related Protocols},
  journal = {npj Quantum Information},
  volume  = {3},
  number  = {1},
  pages   = {23},
  year    = {2017},
  doi     = {10.1038/s41534-017-0025-3}
}

@article{beals2001quantum,
  author  = {Robert Beals and Harry Buhrman and Richard Cleve and Michele Mosca and Ronald de Wolf},
  title   = {Quantum Lower Bounds by Polynomials},
  journal = {Journal of the ACM},
  volume  = {48},
  number  = {4},
  pages   = {778--797},
  year    = {2001},
  doi     = {10.1145/502090.502097}
}

@article{simon1997power,
  author  = {Daniel R. Simon},
  title   = {On the Power of Quantum Computation},
  journal = {SIAM Journal on Computing},
  volume  = {26},
  number  = {5},
  pages   = {1474--1483},
  year    = {1997},
  doi     = {10.1137/S0097539796298637}
}

@inproceedings{childs2003exponential,
  author    = {Andrew M. Childs and Richard Cleve and Enrico Deotto and Edward Farhi and Sam Gutmann and Daniel A. Spielman},
  title     = {Exponential Algorithmic Speedup by a Quantum Walk},
  booktitle = {Proceedings of the Thirty-Fifth Annual ACM Symposium on Theory of Computing},
  pages     = {59--68},
  year      = {2003},
  doi       = {10.1145/780542.780552}
}

@inproceedings{xu2025fattree,
  author    = {Xu, Shifan and Lu, Alvin and Ding, Yongshan},
  title     = {Fat-Tree {QRAM}: A High-Bandwidth Shared Quantum Random Access Memory for Parallel Queries},
  booktitle = {Proceedings of the 30th ACM International Conference on Architectural Support for Programming Languages and Operating Systems, Volume 2},
  series    = {ASPLOS '25},
  year      = {2025},
  pages     = {74--90},
  numpages  = {17},
  location  = {Rotterdam, Netherlands},
  publisher = {Association for Computing Machinery},
  address   = {New York, NY, USA},
  doi       = {10.1145/3676641.3716256}
}

@article{weiss2024faulty,
  author        = {D. K. Weiss and Shifan Xu and Shruti Puri and Yongshan Ding and S. M. Girvin},
  title         = {Faulty Towers: Recovering a Functioning Quantum Random Access Memory in the Presence of Defective Routers},
  journal       = {arXiv preprint arXiv:2411.15612},
  year          = {2024},
  eprint        = {2411.15612},
  archivePrefix = {arXiv},
  primaryClass  = {quant-ph}
}

@article{wang2025transmon,
  author  = {Zhaoyou Wang and Hong Qiao and Andrew N. Cleland and Liang Jiang},
  title   = {Quantum Random Access Memory with Transmon-Controlled Phonon Routing},
  journal = {Physical Review Letters},
  volume  = {134},
  number  = {21},
  pages   = {210601},
  year    = {2025},
  doi     = {10.1103/PhysRevLett.134.210601}
}

@inproceedings{goldreich1987oram,
  author    = {Oded Goldreich},
  title     = {Towards a Theory of Software Protection and Simulation by Oblivious {RAM}s},
  booktitle = {Proceedings of the Nineteenth Annual ACM Symposium on Theory of Computing},
  pages     = {182--194},
  year      = {1987},
  doi       = {10.1145/28395.28416}
}

@article{chor1998pir,
  author  = {Benny Chor and Oded Goldreich and Eyal Kushilevitz and Madhu Sudan},
  title   = {Private Information Retrieval},
  journal = {Journal of the ACM},
  volume  = {45},
  number  = {6},
  pages   = {965--981},
  year    = {1998},
  doi     = {10.1145/293347.293350}
}

@article{jakobi2011private,
  author  = {Markus Jakobi and Christoph Simon and Nicolas Gisin and Cyril Branciard and Jean-Daniel Bancal and Nino Walenta and Hugo Zbinden},
  title   = {Practical Private Database Queries Based on a Quantum-Key-Distribution Protocol},
  journal = {Physical Review A},
  volume  = {83},
  number  = {2},
  pages   = {022301},
  year    = {2011},
  doi     = {10.1103/PhysRevA.83.022301}
}

@article{li2023robust,
  author  = {Li, Zihao and Zhu, Huangjun and Hayashi, Masahito},
  title   = {Robust and Efficient Verification of Graph States in Blind Measurement-Based Quantum Computation},
  journal = {npj Quantum Information},
  volume  = {9},
  pages   = {115},
  year    = {2023},
  doi     = {10.1038/s41534-023-00783-9}
}

@article{bourdoncle2025practical,
  author  = {Bourdoncle, Boris and Emeriau, Pierre-Emmanuel and Hilaire, Paul and Mansfield, Shane and Music, Luka and Wein, Stephen},
  title   = {Towards Practical Secure Delegated Quantum Computing with Semi-Classical Light},
  journal = {Quantum},
  volume  = {9},
  pages   = {1943},
  year    = {2025},
  doi     = {10.22331/q-2025-12-12-1943}
}

@inproceedings{ren2015ring,
  author    = {Ren, Ling and Fletcher, Christopher W. and Kwon, Albert and Stefanov, Emil and Shi, Elaine and van Dijk, Marten and Devadas, Srinivas},
  title     = {Constants Count: Practical Improvements to Oblivious {RAM}},
  booktitle = {Proceedings of the 24th USENIX Security Symposium},
  pages     = {415--430},
  year      = {2015}
}

@inproceedings{mishra2018oblix,
  author    = {Mishra, Pratyush and Poddar, Rishabh and Chen, Jerry and Chiesa, Alessandro and Popa, Raluca Ada},
  title     = {Oblix: An Efficient Oblivious Search Index},
  booktitle = {Proceedings of the 2018 IEEE Symposium on Security and Privacy},
  pages     = {279--296},
  year      = {2018},
  doi       = {10.1109/SP.2018.00045}
}

@inproceedings{sasy2018zerotrace,
  author    = {Sasy, Sajin and Gorbunov, Sergey and Fletcher, Christopher W.},
  title     = {{ZeroTrace}: Oblivious Memory Primitives from {Intel SGX}},
  booktitle = {Proceedings of the Network and Distributed System Security Symposium},
  year      = {2018}
}

@article{kerenidis2004qspir,
  author  = {Kerenidis, Iordanis and de Wolf, Ronald},
  title   = {Quantum Symmetrically-Private Information Retrieval},
  journal = {Information Processing Letters},
  volume  = {90},
  number  = {3},
  pages   = {109--114},
  year    = {2004},
  doi     = {10.1016/j.ipl.2004.02.003}
}

@inproceedings{song2000encryptedsearch,
  author    = {Song, Dawn Xiaodong and Wagner, David and Perrig, Adrian},
  title     = {Practical Techniques for Searches on Encrypted Data},
  booktitle = {Proceedings of the 2000 IEEE Symposium on Security and Privacy},
  pages     = {44--55},
  year      = {2000},
  doi       = {10.1109/SECPRI.2000.848445}
}

@inproceedings{curtmola2006sse,
  author    = {Curtmola, Reza and Garay, Juan and Kamara, Seny and Ostrovsky, Rafail},
  title     = {Searchable Symmetric Encryption: Improved Definitions and Efficient Constructions},
  booktitle = {Proceedings of the 13th ACM Conference on Computer and Communications Security},
  pages     = {79--88},
  year      = {2006},
  doi       = {10.1145/1180405.1180417}
}

@inproceedings{boneh2013quantum,
  author    = {Boneh, Dan and Zhandry, Mark},
  title     = {Secure Signatures and Chosen Ciphertext Security in a Quantum Computing World},
  booktitle = {Advances in Cryptology -- CRYPTO 2013},
  series    = {Lecture Notes in Computer Science},
  volume    = {8043},
  pages     = {361--379},
  year      = {2013},
  publisher = {Springer}
}

@article{luby1988prp,
  author  = {Luby, Michael and Rackoff, Charles},
  title   = {How to Construct Pseudorandom Permutations from Pseudorandom Functions},
  journal = {SIAM Journal on Computing},
  volume  = {17},
  number  = {2},
  pages   = {373--386},
  year    = {1988},
  doi     = {10.1137/0217022}
}

@inproceedings{ambainis2000private,
  author    = {Ambainis, Andris and Mosca, Michele and Tapp, Alain and de Wolf, Ronald},
  title     = {Private Quantum Channels},
  booktitle = {Proceedings of the 41st Annual IEEE Symposium on Foundations of Computer Science},
  pages     = {547--553},
  year      = {2000},
  doi       = {10.1109/SFCS.2000.892142}
}

@article{shriram2023ruffle,
  author  = {Shriram A, Pranav and Koti, Nishat and Kukkala, Varsha Bhat and Patra, Arpita and Gopal, Bhavish Raj and Sangal, Somya},
  title   = {Ruffle: Rapid 3-Party Shuffle Protocols},
  journal = {Proceedings on Privacy Enhancing Technologies},
  volume  = {2023},
  number  = {3},
  pages   = {24--42},
  year    = {2023},
  doi     = {10.56553/popets-2023-0068}
}

@inproceedings{song2024secretshuffle,
  author    = {Song, Xiangfu and Yin, Dong and Bai, Jianli and Dong, Changyu and Chang, Ee-Chien},
  title     = {Secret-Shared Shuffle with Malicious Security},
  booktitle = {Proceedings of the Network and Distributed System Security Symposium},
  year      = {2024},
  doi       = {10.14722/ndss.2024.24021}
}

@inproceedings{gentry2009fhe,
  author    = {Gentry, Craig},
  title     = {Fully Homomorphic Encryption Using Ideal Lattices},
  booktitle = {Proceedings of the 41st Annual ACM Symposium on Theory of Computing},
  pages     = {169--178},
  year      = {2009},
  doi       = {10.1145/1536414.1536440}
}

@inproceedings{ohrimenko2014melbourne,
  author    = {Ohrimenko, Olga and Goodrich, Michael T. and Tamassia, Roberto and Upfal, Eli},
  title     = {The Melbourne Shuffle: Improving Oblivious Storage in the Cloud},
  booktitle = {Automata, Languages, and Programming},
  series    = {Lecture Notes in Computer Science},
  volume    = {8573},
  pages     = {556--567},
  year      = {2014},
  publisher = {Springer},
  doi       = {10.1007/978-3-662-43951-7_47}
}

@inproceedings{patel2018cacheshuffle,
  author    = {Patel, Sarvar and Persiano, Giuseppe and Yeo, Kevin},
  title     = {{CacheShuffle}: A Family of Oblivious Shuffles},
  booktitle = {45th International Colloquium on Automata, Languages, and Programming},
  series    = {Leibniz International Proceedings in Informatics},
  volume    = {107},
  pages     = {161:1--161:13},
  year      = {2018},
  publisher = {Schloss Dagstuhl -- Leibniz-Zentrum f{\"u}r Informatik},
  doi       = {10.4230/LIPIcs.ICALP.2018.161}
}

@inproceedings{banerjee2014spring,
  author    = {Banerjee, Abhishek and Brenner, Hai and Leurent, Ga{\"e}tan and Peikert, Chris and Rosen, Alon},
  title     = {{SPRING}: Fast Pseudorandom Functions from Rounded Ring Products},
  booktitle = {Fast Software Encryption},
  series    = {Lecture Notes in Computer Science},
  volume    = {8424},
  pages     = {38--57},
  year      = {2014},
  publisher = {Springer},
  doi       = {10.1007/978-3-662-46706-0_3}
}

@inproceedings{banerjee2012prf,
  author    = {Banerjee, Abhishek and Peikert, Chris and Rosen, Alon},
  title     = {Pseudorandom Functions and Lattices},
  booktitle = {Advances in Cryptology -- EUROCRYPT 2012},
  series    = {Lecture Notes in Computer Science},
  volume    = {7237},
  pages     = {719--737},
  year      = {2012},
  publisher = {Springer},
  doi       = {10.1007/978-3-642-29011-4_42}
}

@article{lyubashevsky2013ringlwe,
  author  = {Lyubashevsky, Vadim and Peikert, Chris and Regev, Oded},
  title   = {On Ideal Lattices and Learning with Errors over Rings},
  journal = {Journal of the ACM},
  volume  = {60},
  number  = {6},
  articleno = {43},
  pages   = {43:1--43:35},
  year    = {2013},
  doi     = {10.1145/2535925}
}

@article{chillotti2020tfhe,
  title={TFHE: Fast Fully Homomorphic Encryption Over the Torus: I. Chillotti et al.},
  author={Chillotti, Ilaria and Gama, Nicolas and Georgieva, Mariya and Izabach{\`e}ne, Malika},
  journal={Journal of Cryptology},
  volume={33},
  number={1},
  pages={34--91},
  year={2020},
  publisher={Springer}
}

@article{brakerski2014leveled,
  title={(Leveled) fully homomorphic encryption without bootstrapping},
  author={Brakerski, Zvika and Gentry, Craig and Vaikuntanathan, Vinod},
  journal={ACM Transactions on Computation Theory (TOCT)},
  volume={6},
  number={3},
  pages={1--36},
  year={2014},
  publisher={ACM New York, NY, USA}
}

@inproceedings{haah2016sample,
  title={Sample-optimal tomography of quantum states},
  author={Haah, Jeongwan and Harrow, Aram W and Ji, Zhengfeng and Wu, Xiaodi and Yu, Nengkun},
  booktitle={Proceedings of the forty-eighth annual ACM symposium on Theory of Computing},
  pages={913--925},
  year={2016}
}

@inproceedings{o2016efficient,
  title={Efficient quantum tomography},
  author={O'Donnell, Ryan and Wright, John},
  booktitle={Proceedings of the forty-eighth annual ACM symposium on Theory of Computing},
  pages={899--912},
  year={2016}
}

@inproceedings{gagliardoni2017orams,
  title={ORAMs in a quantum world},
  author={Gagliardoni, Tommaso and Karvelas, Nikolaos P and Katzenbeisser, Stefan},
  booktitle={International Workshop on Post-Quantum Cryptography},
  pages={406--425},
  year={2017},
  organization={Springer}
}
